\documentclass[10pt,twocolumn,twoside]{IEEEtran}

\usepackage{amsmath,amssymb,amsthm}
\usepackage{algorithmic}
\usepackage{amsmath,graphicx}
\usepackage{pgf,tikz,pgfplots}
\usepackage{bm}
\usepackage{color}
\usepackage{float}
\usepackage{relsize}
\usepackage{exscale}

\usepackage{array}
\usepackage{eqparbox}
\renewcommand\algorithmiccomment[1]{%
	~\hfill$\triangleright$~\eqparbox{COMMENT}{#1}%
}

\usepackage{subfigure}
\usepackage{cases}
\usepackage{tabularx}

\usepackage{multicol}

\usepackage{hyperref}
\hypersetup{
	colorlinks=true,
	linkcolor=blue,
	filecolor=magenta,
	urlcolor=cyan,
	citecolor=blue
}
\usetikzlibrary{arrows,decorations.markings,automata}
\usepackage{caption}
\usepackage{xspace}
\usepackage{cite}

\usepackage{makecell}
\usepackage{tabularx}
\usepackage{array}

\newtheorem{definition}{Definition}

\newcommand{\CL}{\mathcal{L}}
\newcommand{\V}{\mathcal{V}}
\newcommand{\I}{\mathcal{I}}
\newcommand{\F}{\mathcal{F}}
\newcommand{\K}{\mathcal{K}}
\newcommand{\C}{\mathcal{C}}

\newcommand{\CP}{\mathcal{P}}

\DeclareMathOperator{\conv}{conv}

\def\longpaper{1}

\newcommand{\bx}{\boldsymbol{x}}

\newcommand{\br}{\boldsymbol{r}}

\newcommand{\by}{\boldsymbol{y}}
\newcommand{\bz}{\boldsymbol{z}}
\newcommand{\btheta}{\boldsymbol{\theta}}

\newcommand{\network}{fish-7-10}
\newcommand{\figDIR}{figs}
\newcommand{\figuresize}{1.6in}

\def\x{\boldsymbol{x}}
\def\y{\boldsymbol{y}}

\def\z{\boldsymbol{z}}

\def\r{\boldsymbol{r}}
\def\btheta{\boldsymbol{\theta}}

\newcommand{\rev}[1]{{\color{black}{#1}}}
\newcommand{\LP}{\text{LP}\xspace}
\newcommand{\NLP}{\text{NLP}\xspace}
\newcommand{\CG}{\text{CG}\xspace}
\newcommand{\cCG}{\text{cCG}\xspace}
\newcommand{\vcCG}{\text{cCG'}\xspace}
\newcommand{\EXACT}{\text{EXACT}\xspace}

\newcommand{\LPoR}{\text{LPoR}\xspace}
\newcommand{\LPdR}{\text{LPdR}\xspace}

\newcommand{\NS}{\text{NS}\xspace}
\newcommand{\MILP}{\text{MILP}\xspace}
\newcommand{\MINLP}{\text{MINLP}\xspace}

\renewcommand{\arraystretch}{1.1}

\definecolor{ududff}{rgb}{0.30196078431372547,0.30196078431372547,1.}
\usepackage{algorithm,algorithmic}
\usepackage{cleveref}
\newtheorem{theorem}{Theorem}

\newtheorem{lemma}[theorem]{Lemma}

\usepackage{todonotes}

\usepackage{tabularx}

\crefname{theorem}{Theorem}{Theorems}
\crefname{example}{Example}{Examples}
\crefname{observation}{Observation}{Observations}
\crefname{remark}{Remark}{Remarks}
\crefname{proposition}{Proposition}{Propositions}
\crefname{lemma}{Lemma}{Lemmas}
\crefname{corollary}{Corollary}{Corollaries}
\crefname{fact}{Fact}{Facts}
\crefname{algocf}{Algorithm}{Algorithms}	
\crefname{table}{Table}{Tables}	
\crefname{figure}{Fig.}{Figs.}
\crefname{algorithm}{Algorithm}{Algorithms}
\crefname{section}{Section}{Sections}
\newcolumntype{C}[1]{>{\centering\arraybackslash}m{#1}}

\makeatletter
\newcommand\Label[1]{&\refstepcounter{equation}(\theequation)\ltx@label{#1}&}
\makeatother

\begin{document}
\title{QoS-Aware and Routing-Flexible Network Slicing for Service-Oriented Networks\thanks{Part of this work \cite{Chen2022} has been presented at the 2022 IEEE International Conference on Acoustics, Speech and Signal Processing  (ICASSP), Singapore, Singapore, 22–27 May, 2022.}}
\author{\IEEEauthorblockN{Wei-Kun Chen, Ya-Feng Liu, Yu-Hong Dai, and Zhi-Quan Luo}
	\thanks{W.-K. Chen is with the School of Mathematics and Statistics/Beijing Key Laboratory on MCAACI, Beijing Institute of Technology, Beijing 100081, China (e-mail: chenweikun@bit.edu.cn).
		\rev{Y.-F. Liu is with the Ministry of Education Key Laboratory of Mathematics and Information Networks, School of Mathematical Sciences, Beijing
		University of Posts and Telecommunications, Beijing 102206, China (email:
		yafengliu@bupt.edu.cn).}
		{Y.-H. Dai} is with the State Key Laboratory
		of Scientific and Engineering Computing, Institute of Computational Mathematics and Scientific/Engineering Computing, Academy of Mathematics and Systems Science, Chinese Academy of Sciences, Beijing 100190, China (e-mail:  dyh@lsec.cc.ac.cn).
		Z.-Q. Luo is with the Shenzhen Research Institute of Big Data and The Chinese University of Hong Kong, Shenzhen 518172, China (e-mail: luozq@cuhk.edu.cn).
	}
}

\maketitle

\begin{abstract}
	In this paper, we consider the network slicing (\NS) problem which {aims to} map multiple customized virtual network requests (also called services) to a common shared network infrastructure and manage network resources to meet diverse quality of service (QoS) requirements.
	We propose a mixed-integer nonlinear programming (MINLP) formulation for the considered NS problem that can flexibly route the traffic flow of the services on multiple paths and provide end-to-end delay and reliability guarantees for all services. 
	To overcome the computational difficulty due to the intrinsic nonlinearity in the MINLP formulation, 
    we transform the \MINLP formulation into an equivalent mixed-integer linear programming
	(\MILP) formulation and further show that their continuous relaxations are equivalent.
	In sharp contrast to the continuous relaxation of the \MINLP formulation which is a nonconvex nonlinear programming problem, the continuous relaxation of the \MILP formulation is a polynomial-time solvable linear programming problem, which significantly facilitates the algorithmic design.
	Based on the newly proposed \MILP formulation, we develop a customized column generation (\cCG) algorithm for solving the \NS problem.
	The proposed \cCG algorithm is a decomposition-based  algorithm and is particularly suitable for solving large-scale \NS problems.
	Numerical results demonstrate the efficacy of the proposed formulations and the proposed \cCG algorithm.
\end{abstract}

\section{Introduction}
\label{sec:intro}

Network function virtualization (NFV) is one of the key technologies for the fifth generation (5G) and next-generation 6G communication networks \cite{Mijumbi2016,Vassilaras2017}.
Different from traditional networks in which service requests (e.g., high-dimensional video, remote robotic surgery, autonomous driving, and machine control) are implemented by dedicated hardware in fixed locations, NFV-enabled networks efficiently leverage virtualization technologies to configure some specific cloud nodes in the network to process network service functions on-demand, and then establish a customized virtual network for all service requests. 
However, since virtual network functions (VNFs) of all services are implemented over a single shared cloud
network infrastructure, it is crucial to efficiently allocate network (e.g., computing and communication) resources subject to diverse quality of service (QoS) requirements of all services and capacity constraints of all cloud nodes and links in the network.
This resource allocation problem in the service-oriented network is called \emph{network slicing} (\NS) in the literature, which jointly considers the VNF placement (that maps VNFs into cloud nodes in the network) and the traffic routing (that finds paths connecting two cloud nodes which perform two adjacent VNFs in the network).
This paper focuses on the mathematical optimization formulation and the solution approach of the \NS problem.

\rev{When establishing problem formulations and developing solution approaches for the \NS problem, two key factors,  that well reflect the practical application of the NS problem,  must be considered.
First, many services are mission-critical and have very stringent QoS requirements, including the E2E delay and reliability requirements. 
For example, to guarantee robust real-time interactions between the surgeon and the remote surgical robot, tele-surgery services require an E2E delay of up to 1 ms and an E2E reliability of up to $99.9999999\%$ \cite{Kim2019}.
Therefore, to ensure guaranteed E2E delay and reliability of the services, it is important to explicitly enforce the related constraints in the problem formulations.
The works 	\cite{Chen2021a,Chen2023,Promwongsa2020,Addis2015,Oljira2017,Luizelli2015,Woldeyohannes2018,Mohammadkhan2015,Vizarreta2017,Qu2017,Li2024,Dobreff2025} developed mathematical formulations that explicitly enforce the E2E delay constraints of the services
while the works \cite{Guerzoni2014,Yeow2010,Karimzadeh-Farshbafan2020,Vizarreta2017,Qu2017,Li2024,Dobreff2025} developed mathematical formulations that explicitly enforce the E2E reliability constraints of the services.}
\rev{Second, as the communication resource (i.e., link capacity) is limited, a flexible traffic routing strategy is usually needed  to ensure a high-quality solution of the NS problem.
In the literature, to simplify the solution approaches,  various simplified traffic routing strategies have been proposed, including predetermined path routing  (i.e., selecting the paths from a predetermined path set) \cite{Luizelli2015,Jiang2012,Guo2011,Jarray2012},
limited number of hops routing \cite{Jarray2015}, 
and single-path routing \cite{Mijumbi2015,Liu2017,Addis2015,Woldeyohannes2018,Vizarreta2017,Oljira2017,Mohammadkhan2015,Qu2017,Li2024,Dobreff2025}. 
The works \cite{Zhang2017,Chowdhury2012,Chen2023b,Chen2021a,Chen2023,Promwongsa2020} developed formulations that allow the traffic flow to be distributed across arbitrary paths within the underlying network, thereby fully utilizing the flexibility of traffic routing.}

\begin{table}
	\renewcommand{\arraystretch}{1.3}
	\setlength{\tabcolsep}{3pt}
	\centering
	\caption{A summary of state-of-the-art mathematical optimization formulations of the \NS problem.}
	\label{works}
	\begin{tabular}{|c|c|c|c|c|}
		\hline
		{Works}& \makecell{Flexible\\routing} &  \makecell{Guaranteed\\E2E Delay} & \makecell{Guaranteed\\E2E Reliability}
		\\\hline
		\cite{Zhang2017,Chowdhury2012,Chen2023b} & \checkmark &-- & -- 
		\\\hline
		\cite{Chen2021a,Chen2023,Promwongsa2020} & \checkmark &\checkmark & --
		\\\hline
		\cite{Addis2015,Oljira2017,Luizelli2015,Woldeyohannes2018,Mohammadkhan2015} & -- &\checkmark & --
		\\\hline
		\makecell{\cite{Guerzoni2014,Yeow2010,Karimzadeh-Farshbafan2020}} & -- &-- & \checkmark
		\\\hline
		\cite{Vizarreta2017,Qu2017,Li2024,Dobreff2025} & -- &\checkmark & \checkmark
		\\\hline
		\makecell{\cite{Mijumbi2015,Liu2017,Jiang2012,Guo2011,Jarray2012,Jarray2015}}& -- &-- & --
		\\\hline
		\makecell{Our work}&   \checkmark & \checkmark &  \checkmark \\\hline
	\end{tabular}
	\vspace*{-0.4cm}
\end{table}

In \cref{works}, we summarize various mathematical optimization formulations of the \NS problem. 
These formulations are different in terms of their routing strategies and E2E QoS guarantees.
As can be clearly seen from \cref{works}, for the \NS problem, none of the existing formulations/works  simultaneously takes all of the above practical factors (\rev{i.e.}, flexible routing and E2E delay and reliability requirements) into consideration.
The motivation of this paper is to provide a mathematical optimization formulation of the \NS problem that simultaneously allows the traffic flows to be flexibly transmitted on (possibly) multiple paths and satisfies the E2E delay and reliability requirements of all services, and to develop an efficient algorithm for solving it.


\subsection{Our Contributions}

In this paper, we first propose a QoS-aware and routing-flexible formulation for the \NS problem and then develop an efficient algorithm for solving the formulated problem. 
The main contributions of this paper are summarized as follows.

\begin{itemize}
	\item \emph{A new MINLP formulation}: 
	We first propose a mixed-integer nonlinear programming  (\MINLP) formulation for the \NS problem, 
	which minimizes network resource consumption subject to the capacity constraints of all cloud nodes and links and the E2E delay and reliability constraints of all services.
	Compared with the existing state-of-the-art  formulations  in \cite{Chen2023} and \cite{Vizarreta2017}, the proposed formulation allows to flexibly route the traffic flow of the services on multiple paths and provides E2E delay and reliability guarantees for all services.
	\item 
	\emph{A computationally efficient MILP reformulation}: Due to the intrinsic nonlinearity, the MINLP formulation appears to be computationally difficult. 
	To overcome this difficulty, we transform the MINLP formulation into an equivalent  mixed-integer linear programming (\MILP) formulation and show that their continuous relaxations are equivalent (in terms of sharing the same optimal values).
	However, different from the continuous relaxation of the \MINLP formulation which is a nonconvex nonlinear programming (\NLP) problem, the continuous relaxation of the \MILP formulation is a polynomial-time solvable linear programming (\LP) problem, which is amenable to the algorithmic design.
	The key to the successful transformation is a novel way of rewriting \emph{nonlinear} flow conservation constraints for traffic routing as \emph{linear} constraints.
	\item \emph{An efficient  \cCG algorithm}: We develop an {efficient} customized \CG (\cCG) algorithm based on the \MILP formulation. 
	The proposed \cCG algorithm is a two-stage algorithm: the first stage decomposes the large-scale \NS problem into \LP and \MILP subproblems of much smaller sizes, and iteratively solves them to collect effective embedding patterns for each service;
	the second stage finds a high-quality solution for the \NS problem by solving a small \MILP problem based on embedding patterns obtained in the first stage.
	Theoretically, the  proposed \cCG algorithm can be interpreted as the \LP relaxation rounding algorithm equipped with a strong \LP relaxation and an optimal rounding strategy \cite{Berthold2013}. 
	This theoretical interpretation, together with its decomposition nature, makes the proposed \cCG algorithm well-suited for tackling large-scale \NS problems and for finding high-quality solutions. 
\end{itemize}

Extensive computational results demonstrate the effectiveness of the proposed formulations and the efficiency of the proposed \cCG algorithm.
More specifically, our computational results show: (i) the proposed \MINLP/\MILP formulations are more effective than the existing formulations in \cite{Chen2023} and \cite{Vizarreta2017} in terms of the solution quality; 
(ii) the proposed \MILP formulation is much more computationally efficient than the \MINLP formulation; 
(iii) compared with calling general-purpose \MILP solvers (e.g., CPLEX) for solving the \MILP formulation \rev{and state-of-the-art algorithms in \cite{Chen2023} and \cite{Chowdhury2012}}, the proposed \cCG algorithm 
can find high-quality solutions for large-scale \NS problems within significantly less computational time.

{In our previous work \cite{Chen2022}, we presented the \MILP and \MINLP formulations for solving the \NS problem,
which corresponds to \cref{sect:modelformulation} and part of \cref{subsect:MILP} in this paper. 
Compared to its conference version \cite{Chen2022}, 
this paper provides a rigorous theoretical analysis on the relation between the \MILP and \MINLP formulations, which provides important insights into why the performance of solving the \MILP formulation is better than that of solving the  \MINLP formulation.
More importantly, this paper develops an efficient \cCG algorithm for solving large-scale \NS problems, which is completely new.}

This paper is organized as follows. 
\rev{\cref{sect:literature} reviews the related works.}
\cref{sect:modelformulation} introduces the system model and provides an \MINLP formulation for the \NS problem.
\cref{sect:MILP} transforms the \MINLP formulation into an equivalent \MILP formulation. 
\cref{sect:CG} develops a \cCG algorithm based on the newly proposed \MILP formulation.
\cref{sect:numres} reports the computational results. 
Finally, \cref{sect:conclusion} draws the conclusion.

\section{Literature Review}
\label{sect:literature}
In recent years, there are considerable works on the \NS problem and its variants; see \cite{Zhang2017,Chowdhury2012,Mijumbi2015,Liu2017,Chen2023b,Chen2021a,Addis2015,Chen2023,Promwongsa2020,Luizelli2015,Woldeyohannes2018,Guerzoni2014,Yeow2010,Jiang2012,Guo2011,Jarray2015,Jarray2012,Vizarreta2017,Oljira2017,Mohammadkhan2015,Qu2017,Karimzadeh-Farshbafan2020,Li2024,Dobreff2025} 
and the references therein.
In particular, for the \NS problem with a limited network resource constraint, the works \cite{Zhang2017,Chowdhury2012}  proposed \LP relaxation rounding algorithms to find a feasible solution of the formulated \MILP formulation; 
the works \cite{Mijumbi2015,Liu2017} developed column generation (\CG) algorithms to find a feasible solution;
the work \cite{Chen2023b} proposed a Benders decomposition  algorithm to find a global optimal solution.
However, the formulations in  \cite{Zhang2017,Chowdhury2012,Mijumbi2015,Liu2017,Chen2023b}  neither considered E2E delay nor E2E reliability constraints of the services, which are two key design considerations in the shared cloud network infrastructure \cite{Jin2024,Agyapong2014}.
The works \cite{Chen2021a,Addis2015,Oljira2017,Chen2023,Promwongsa2020,Luizelli2015,Woldeyohannes2018,Mohammadkhan2015}  
incorporated the E2E delay constraints of the services into their formulations and solved the problem by off-the-shelf \MILP solvers \cite{Chen2021a,Addis2015,Oljira2017} and heuristic algorithms such as \LP relaxation \rev{dynamic} rounding algorithm \cite{Chen2023}, Tabu search \cite{Promwongsa2020}, binary search \cite{Luizelli2015}, greedy heuristic \cite{Mohammadkhan2015}, and the so-called ClusPR algorithm \cite{Woldeyohannes2018}.
The works \cite{Guerzoni2014,Yeow2010,Karimzadeh-Farshbafan2020} took into account of the E2E reliability constraints of the services in their formulations and solved the problem using off-the-shelf \MILP solvers. 
However, the aforementioned works still either did not consider the E2E reliability constraints of the services (e.g., \cite{Chen2021a,Addis2015,Oljira2017,Chen2023,Promwongsa2020,Luizelli2015,Mohammadkhan2015,Woldeyohannes2018}) 
or did not consider the E2E delay constraints of the services (e.g., \cite{Guerzoni2014,Yeow2010,Karimzadeh-Farshbafan2020}).
Obviously, formulations that do not consider E2E delay or E2E reliability constraints may return a solution that violates these important QoS requirements.

For the traffic routing (that finds paths connecting two cloud nodes) of the \NS problem, 
the works \cite{Luizelli2015,Jiang2012,Guo2011,Jarray2012} simplified {the routing strategy by selecting paths from a predetermined set of paths}, 
the works \cite{Jarray2015} limited the maximum number of hops in the flows to be a prespecified number, and the works \cite{Mijumbi2015,Liu2017,Addis2015,Woldeyohannes2018,Vizarreta2017,Oljira2017,Mohammadkhan2015,Qu2017,Li2024} enforced that the data flow should be transmitted over only a single path.
{Apparently, formulations that are based on such restrictive assumptions do not fully utilize the flexibility of traffic routing (i.e., allowing the traffic flow to be distributed across arbitrary paths within the underlying network), thereby leading to a suboptimal performance of the whole network \cite{Yu2008,Chen2021a}.}

The works \rev{\cite{Guerzoni2014,Yeow2010,Qu2017,Karimzadeh-Farshbafan2020,Li2024,Dobreff2025}} proposed different \rev{protection (or redundancy)  techniques} (which reserve additional cloud node or link capacities to provide resiliency against links' or  cloud nodes' failure) to \rev{improve} the E2E reliability of the services.
\rev{However, as redundant cloud and communication resources are implemented for the services,}
the protection schemes generally lead to inefficiency of resource allocation \cite{Jin2024}.

\section{System Model and Problem Formulation}
\label{sect:modelformulation}
We use \rev{a directed graph $\mathcal{G} = (\mathcal{I}, \mathcal{L})$} to represent the substrate network, where $\mathcal{I}=\{i\}$ and $\mathcal{L}=\{(i,j)\}$ denote the sets of nodes and \rev{directed} links, respectively.
Let $ \mathcal{V} \subseteq \mathcal{I}$ be the set of cloud nodes.
Each cloud node $ v $ has a computational capacity $ \mu_v $ and a reliability $\gamma_{v}$ \cite{Vizarreta2017,Guerzoni2014}.
{We follow \cite{Zhang2017} to assume that processing one unit of data rate consumes one unit of (normalized) computational capacity.}
Each link $ (i,j) $ has an expected (communication) delay $ d_{ij} $ \cite{Woldeyohannes2018}, a reliability $\gamma_{ij}$ \cite{Vizarreta2017,Guerzoni2014}, and  a total data rate upper bounded by the capacity $C_{ij}$.
A set of services $\mathcal{K}$ is needed to be supported by the network.
Let $S^k,D^k\notin \mathcal{V}$ be the source and destination nodes of service $k\in \K$.
Each service $ k $ relates to a customized service function chain (SFC) consisting of $ \ell_k $ service functions that have to be processed in sequence by the network: $f_{1}^k\rightarrow f_{2}^k\rightarrow \cdots \rightarrow f_{\ell_k}^k$ \cite{Zhang2013,Halpern2015,Mirjalily2018}.
As required in \cite{Zhang2017,Woldeyohannes2018}, in order to minimize the coordination overhead, each function must be processed at exactly one cloud node.
If function $ f^k_s $, $ s \in \mathcal{F}^k := \{1,2,\ldots, \ell_k\} $, is processed by cloud node $ v $ in $ \mathcal{V} $, the expected NFV delay is assumed to be known as $d_{v}^{k,s} $, which includes both processing and queuing delays \cite{Luizelli2015,Woldeyohannes2018}.
For service $k$, \rev{the data rate at which no function has been processed by any cloud node is denoted as $\lambda_0^k$; similarly, the data rate at which function $f_s^k$  has just been processed (and function $f_{s+1}^{k}$ has not been processed) by some cloud node is denoted as $\lambda_{s}^k$.}
Each service $ k $ has an E2E delay requirement and an E2E reliability requirement, denoted as $ \Theta^k $ and $\Gamma^k$, respectively.
\cref{notation} summarizes the notations in the \NS problem.

	\begin{table}[t]
	\caption{Summary of notations in the \NS problem.}
	\label{notation}
	\centering	
	\setlength{\tabcolsep}{3pt} 
	\renewcommand{\arraystretch}{1.2} 
	\begin{tabularx}{0.48\textwidth}{|c|X|}
			\hline
			\multicolumn{2}{|c|}{Parameters}   \\
			\hline 
			$ \mu_v $                        & computational capacity of  cloud node $ v $
			\\
			\hline
			$\gamma_v $                 & reliability of cloud node $v$\\
			\hline 
			$ C_{ij} $                      & communication capacity of link $ (i,j) $                                                                                                                                                                                     \\
			\hline
			$ d_{ij} $                      & communication delay of link $ (i,j) $                                                                                                               
			\\
			\hline 
			$ \gamma_{ij} $                      &  reliability of link $ (i,j)$                                                                                         \\
			\hline 
			$ \mathcal{F}^k $               & the index set that corresponds to service $ k $'s SFC                                                                                                                                                                                   \\
			\hline
			$ f_s^k $                        & the $ s $-th function in the SFC of service $ k $                                                                                                                                                                                       \\
			\hline 
			$ d_{v}^{k,s} $                   & NFV delay when function $f_s^k$ is hosted at cloud node $ v $                                                                                                                                                                        \\
			\hline 
			$ \lambda_s^k $                 &  \rev{data rate at which function $f_s^k$  has just been processed and function $f_{s+1}^{k}$ has not been processed}                                                                                                                                                                  \\
			\hline
			$ \Theta^k $                     & E2E latency threshold of service $ k $
			\\
			\hline
			$ \Gamma^k $                     & E2E reliability threshold of service $ k $                                                                                                                                                                                                  \\
			\hline
			\multicolumn{2}{|c|}{Variables}                                                                                                                                                                                                                                       \\
			\hline 
			$ y_v $                          & binary variable indicating whether or not cloud node $ v $ is activated                                                                                                                                                   \\
			\hline
			$ \rev{x_{v}^{k}} $                   & \rev{binary variable indicating whether or not  cloud node $ v $ processes some function(s) in the SFC of service $k$}                                                                                                                                                                                                     \\
			\hline 
			$ x_{v}^{k,s} $                   & binary variable indicating whether or not  cloud node $ v $ processes function $ f_s^k $                                                                                                                                                                                                                                                              \\
			\hline  
			$ r^{k,s, p} $      & data rate on the $p$-th path of flow $(k,s)$  that is used to route the traffic flow between the two cloud nodes hosting functions $ f_s^k $ and $ f_{s+1}^k $, respectively \\
			\hline  
			$ z_{ij}^{k,s,p} $ & binary variable indicating whether or not link $ (i,j) $ is on the $p$-th path of flow $(k,s)$                                                                                         \\
			\hline  
			$ r_{ij}^{k,s,p} $ & data rate on link $ (i,j) $ which is used by the $ p $-th path of flow $(k,s)$                                                                                                         \\
			\hline 
			$ \theta^{k,s} $                  & communication delay due to the traffic flow from the cloud node hosting function $ f_s^k $ to the cloud node hosting function $ f_{s+1}^k $                                                                                          \\
			\hline 
		\end{tabularx}
		\vspace{-0.5cm}
	\end{table}

The \NS problem {involves determining} the VNF placement, the routes, and the associated data rates on the corresponding routes of all services while satisfying the capacity constraints of all cloud nodes and links and the E2E delay and reliability constraints of all services.
Next, we present the constraints and objective function of the problem formulation in detail.
\\[5pt]
{\bf\noindent$\bullet$ VNF placement\\[5pt]}
\indent Let $x_{v}^{k,s}=1$ indicate that function $f^k_s$ is processed by cloud node $v$; otherwise, $x_{v}^{k,s}=0$.
Each function $f_s^k$ must be processed by exactly one cloud node, i.e.,
\begin{eqnarray}
\label{onlyonenode}
\sum_{v\in \mathcal{V}}x_{v}^{k,s}=1,~\forall ~k \in \mathcal{K}, ~s\in  \mathcal{F}^k.
\end{eqnarray}
In addition, we introduce binary variable $x_v^k$ to denote whether there exists at least one function in $\mathcal{F}^k$ processed by cloud node $v$ and binary variable $y_v$ to denote whether cloud node $v$ is activated and powered on.
By definition, we have 
\begin{align}
\label{xyxelation}
& x_{v}^{k,s} \leq  x_{v}^k, ~\forall~v \in \mathcal{V},~k \in \mathcal{K},~s \in \mathcal{F}^k, \\
& x_{v}^k\leq  y_v, ~\forall~v \in \mathcal{V},~k \in \mathcal{K}.\label{xyrelation}
\end{align}
The node capacity constraints can be written as follows:
\begin{equation}
\label{nodecapcons}
\sum_{k\in \mathcal{K}}\sum_{s \in \mathcal{F}^k}\lambda_s^k x_{v}^{k,s}\leq \mu_v y_v,~\forall~ v \in \mathcal{V}.
\end{equation}
{\bf\noindent$\bullet$ Traffic routing\\[5pt]}
\indent Let $ (k,s) $ denote the traffic flow which is routed between the two cloud nodes hosting the two adjacent functions $ f_s^k $ and $ f_{s+1}^k $.
We follow \cite{Chen2021a} to assume that there are at most $P$ paths that can be used to route flow $(k,s)$ and denote $\mathcal{P}=\{1,2, \ldots, P\}$.
%
%
Let binary variable $ z_{ij}^{k,s,p}$ denote whether (or not) link $ (i,j) $ is on the $ p $-th path of flow $ (k,s) $.
Then, to ensure that the functions of each service $k$ are processed in the prespecified order $f_{1}^k\rightarrow f_{2}^k\rightarrow \cdots \rightarrow f_{\ell_k}^k$, we need
\begin{multline}
\sum_{j: (j,i) \in \mathcal{{L}}} z_{ji}^{k, s, p} - \sum_{j: (i,j) \in \mathcal{{L}}} z_{ij}^{k, s,  p}=b_{i}^{k,s}(\boldsymbol{x}),\qquad\\
\forall~ i \in \mathcal{I},~k \in \mathcal{K},~s \in \mathcal{F}^k\cup \{0\},~ p \in \mathcal{P}, \label{SFC1}
\end{multline}
where
\begin{equation*}
b_{i}^{k,s}(\boldsymbol{x})=\left\{\begin{array}{ll}
-1,&\text{if~}s=0~\text{and}~i= S^k;\\[5pt]
x_{i}^{k,s+1},&\text{if~}s=0~\text{and}~i\in \mathcal{V};\\[5pt]
x_{i}^{k,s+1}-x_{i}^{k,s},&\text{if~} 1\leq s < \ell_k 
~\text{and}~i\in \mathcal{V};\\[5pt]
-x_{i}^{k,s},&\text{if~}s=\ell_k
~\text{and}~i\in \mathcal{V};\\[5pt]
1,&\text{if~}s=\ell_k
~\text{and}~i=D^k;\\[5pt]
0,& \text{otherwise}.
\end{array}
\right.
\end{equation*}
Observe that the right-hand side of the above flow conservation constraint \eqref{SFC1} {depends on} whether node $i$ is a source node, a destination node, an activated/inactivated cloud node, or an intermediate node.
Let us elaborate on the second case where $b_{i}^{k,s}(\x)=x_{i}^{k,1}$: if $x_i^{k,1}=0$, then \eqref{SFC1} reduces to the classical flow conservation constraint; if $x_i^{k,1}=1,$ then \eqref{SFC1} reduces to 
$$ \sum_{j: (j,i) \in \mathcal{{L}}} z_{ji}^{k, s, p} - \sum_{j: (i,j) \in \mathcal{{L}}} z_{ij}^{k, s,  p} = 1,$$
which enforces that for flow $(k,0)$, the difference of the numbers of links coming into node $i$ and going out of node $i$ should be equal to $1$.



Next, we follow \cite{Promwongsa2020,Chen2023} to present the flow conservation constraints for the data rates.
To proceed, we need to introduce variable  $ r^{k,s,p}\in [0,1] $ denoting the ratio or the fraction of data rate $\lambda_s^k$ on the $ p $-th path of flow $ (k,s) $ and variable   $ r_{ij}^{k,s,p} \in [0,1]$ denoting the ratio or the fraction of data rate $\lambda_{s}^k$ on link $(i,j)$ (when $z_{ij}^{k,s,p}=1$).
By definition, we have
\begin{align}
	& \sum_{p \in \mathcal{P}}  r^{k,s,p} =  1,   ~ \forall~ k \in \mathcal{K},~s\in \mathcal{F}^k\cup \{0\}, \label{relalambdaandx11}\\
	&r_{ij}^{k, s, p} = r^{k,s,p} z_{ij}^{k, s,p }, \nonumber                                                                                                  \\
	& \quad~~~ \forall~(i,j) \in {\mathcal{L}}, ~k \in \mathcal{K}, ~s \in \mathcal{F}^k\cup \{0\},~p \in \mathcal{P}. \label{nonlinearcons}
\end{align}
Note that constraint \eqref{nonlinearcons} is \emph{nonlinear}.
Finally, the total data rates on link $ (i,j) $ is upper bounded by capacity $ C_{ij} $:
\begin{equation}
\label{linkcapcons1}
\sum_{k \in \mathcal{K}} \sum_{s\in \mathcal{F}^k \cup \{0\}}\sum_{p \in \mathcal{P}} \lambda_{s}^k r_{ij}^{k, s,p} \leq C_{ij}, ~  \forall~(i,j) \in \mathcal{L} .
\end{equation}
{\bf\noindent$\bullet$ E2E reliability\\[5pt]}
\indent To model the reliability of each service, we introduce binary variable $z_{ij}^k$ to denote whether link $(i,j)$ is used to route the traffic flow of service $k$.
By definition, we have 
\begin{multline}
 z_{ij}^{k, s,  p} \leq z_{ij}^{k}, ~                                                                                          
  \forall~(i,j) \in {\mathcal{L}}, ~k \in \mathcal{K}, ~s \in \mathcal{F}^k\cup \{0\},~p \in \mathcal{P}. \label{zzrelcons}
\end{multline}
The E2E reliability of service $k$ is defined as the product of the reliabilities of all cloud nodes hosting all functions $f_s^k$ in its SFC and the reliabilities of all links used by service $k$ \cite{Vizarreta2017}.
The following constraint ensures that the E2E reliability of service $k$ is larger than or equal to its given threshold $\Gamma^k\in[0,1]$:
\begin{equation*}
	\prod_{v\in \mathcal{V}} \rho_{v}^k\cdot 	\prod_{(i,j)\in \mathcal{L}} \rho_{ij}^k \geq \Gamma^k,~\forall~k \in \mathcal{K},
\end{equation*}
where 
\begin{equation*}
\begin{aligned}
\rho_{v}^k=\left\{\begin{array}{ll}
\gamma_v,&\text{if~}x_{v}^k=1;\\[5pt]
1,&\text{if~}x_{v}^k=0;
\end{array}
\right. & ~~
\rho_{ij}^k=\left\{\begin{array}{ll}
\gamma_{ij},&\text{if~}z_{ij}^k=1;\\[5pt]
1,&\text{if~}z_{ij}^k=0.
\end{array}
\right.
\end{aligned}
\end{equation*}
The above nonlinear constraint can be equivalently linearized as follows.
By the definitions of $\rho_v^k$ and $\rho_{ij}^k$, we have
\begin{equation*}
\log (\rho_{v}^k)= \log (\gamma_v) \cdot x_{v}^k  ~\text{and}~ \log (\rho_{ij}^k)=\log (\gamma_{ij}) \cdot z_{ij}^k.
\end{equation*}
Then we can apply the logarithmic transformation on both sides of the above constraint and obtain an equivalent \emph{linear} E2E reliability constraint: 
\begin{multline}\label{E2Ereliability2}
\sum_{v\in \mathcal{V}}  \log (\gamma_v) \cdot x_{v}^k +  	\sum_{(i,j)\in \mathcal{L}} \log (\gamma_{ij}) \cdot z_{ij}^k\\
\geq \log(\Gamma^k),~\forall~k \in \mathcal{K}.
\end{multline}
%
{\bf\noindent$\bullet$ E2E delay\\[5pt]}
\indent We use variable $ \theta^{k,s} $ to denote the communication delay due to the traffic flow from the cloud node hosting function $ f^k_s $ to the cloud node hosting function $ f^k_{s+1} $. 
By definition, \rev{$\theta^{k,s}$ must be the largest delay among the $P$ paths:
	\begin{equation}
		\label{maxdelay1}
		\theta^{k,s} = \max_{p \in  \mathcal{P}}\left\{\sum_{(i,j) \in \mathcal{{L}}}  d_{ij}  z_{ij}^{k, s, p}\right\},~\forall~k \in \mathcal{K}, ~s\in \mathcal{F}^k\cup \{0\}. 
\end{equation}}%
%
To guarantee that service $k$'s {E2E} delay is not larger than its threshold $\Theta^k$, we need the following constraint:
\begin{equation}
\label{delayconstraint}
\theta_\text{N}^k +\theta_\text{L}^k  \leq \Theta^k,~\forall~k \in  \mathcal{K},
\end{equation}
where
\begin{equation*}
	\begin{aligned}
		\theta_\text{N}^k =  \sum_{v \in \mathcal{{V}}}\sum_{s \in \mathcal{F}^k} d_{v}^{k,s} x_{v}^{k,s}~\text{and}~		\theta_\text{L}^k = \sum_{s \in \mathcal{F}^k\cup \{0\}} \theta^{k,s}
	\end{aligned}
\end{equation*}
denote the total NFV delay on the cloud nodes and the total communication delay on the links of service $ k $, respectively.\\[5pt]
{\bf\noindent$\bullet$ Problem formulation\\[5pt]}
\indent The \NS problem is to minimize a weighted sum of the total number of activated nodes (equivalent to the total power consumption in the cloud network \cite{Chen2021a}) and the total link capacity consumption in the whole network:
\begin{align}
 \min_{\substack{\boldsymbol{x},\,\boldsymbol{y},\\\boldsymbol{r},\,\boldsymbol{z},\,\boldsymbol{\theta}}}~   & \sum_{v \in \mathcal{V}}y_v + \sigma \sum_{(i,j)\in \mathcal{L}} \sum_{k\in \mathcal{K}} \sum_{s\in \mathcal{F}^k\cup \{0\}} \sum_{p\in \mathcal{P}} \lambda_s^k r_{ij}^{k,s,p} \nonumber \\
 ~{\text{s.\,t.~~}}                                                                    &   \eqref{onlyonenode}\text{--}\eqref{rbounds}, \label{minlp}
 \tag{\text{MINLP}}
\end{align}
where
\begin{align}
	& x_{v}^{k,s},~x_v^k,~y_v\in\{0,1\},~\forall ~v\in\mathcal{{V}},~k\in\mathcal{K}, ~s\in \mathcal{F}^k,\label{xybounds}\\
	&  r_{ij}^{k, s, p}\geq 0,~z_{ij}^{k, s, p},~z_{ij}^k\in \{0,1\},~\theta^{k,s} \geq 0, \nonumber \\
	& \qquad \qquad ~\forall~(i,j)\in \mathcal{L},~k\in \mathcal{K},~s\in \mathcal{F}^k\cup \{0\},~p \in \mathcal{P}, \label{rzbounds} \\
	& r^{k,s,p} \geq 0,~\forall~k\in \mathcal{K},~s\in \mathcal{F}^k\cup \{0\},~p \in \mathcal{P}.\label{rbounds}
\end{align}
\rev{In the above, $\sigma$ is a small positive constant, implying that the term of the total number of activated cloud nodes is much more important than the link capacity consumption term. Indeed, using a similar argument as in \cite{Chen2021a}, we can show that solving formulation \eqref{minlp} with a sufficiently small positive $\sigma$ is equivalent to solving a two-stage formulation, where the first stage minimizes the total number of activated cloud nodes and with the minimum   total number of activated cloud nodes, the second stage minimizes the total  link capacity consumption.}
The advantage of incorporating the (second) link capacity consumption term into the objective function of  \eqref{minlp} is as follows.
First, it is helpful in avoiding cycles in the traffic flow between the two nodes hosting the two adjacent service functions.
\rev{Second, minimizing the total link capacity consumption
	enables an efficient reservation of more link capacities for future use \cite{Woldeyohannes2018}.
\rev{Third,} minimizing the total data rate in the whole network \rev{may decrease} the \rev{overall} E2E delay and increase the \rev{overall} E2E reliability of all services \rev{as (i) 
	less nodes and links tend to be used by a service, and 
	(ii) minimizing the total link capacity consumption enables an efficient reservation of more link capacities \cite{Woldeyohannes2018}, which allows the traffic flow between the two nodes hosting the two adjacent functions to move more flexibly, leading to a smaller E2E delay and a higher E2E reliability.}} 

\rev{Three remarks on formulation  \eqref{minlp} are in order. First,} 
it is worthwhile remarking that \rev{by rewriting the nonlinear constraints \eqref{nonlinearcons} and \eqref{maxdelay1} as linear constraints,} formulation \eqref{minlp} can be \rev{equivalently} reformulated as an \MILP~\rev{formulation}\rev{. Indeed, constraints \eqref{maxdelay1} can be equivalently linearized as 
	\begin{equation}
		\theta^{k,s} \geq \sum_{(i,j) \in \mathcal{{L}}}  d_{ij}  z_{ij}^{k, s, p},                                                   ~
		\forall~k \in \mathcal{K}, ~s \in \mathcal{F}^k \cup \{0\}, ~p \in \mathcal{P} \label{consdelay2funs1}.
	\end{equation}
	Obviously, the nonlinear constraint \eqref{maxdelay1} implies \eqref{consdelay2funs1}.
	On the other hand, {as variables $\{\theta^{k,s}\}$ do not appear in the objective function}, there always exists an optimal solution $(\boldsymbol{x}, \boldsymbol{y}, \boldsymbol{r}, \boldsymbol{z}, \boldsymbol{\theta})$ of problem \eqref{minlp} (with \eqref{maxdelay1} replaced by \eqref{consdelay2funs1}) 
	such that  \eqref{maxdelay1} holds.
	Each of the nonlinear constraints in \eqref{nonlinearcons}} can be equivalently linearized as

\begin{equation}
\label{bilinearcons}
\tag{7'}
\begin{aligned}
& r_{ij}^{k,s,p} \geq  z_{ij}^{k, s, p} + r^{k,s,p} -1,   
\\
& r_{ij}^{k,s,p} \leq   z_{ij}^{k, s,p} ,~ r_{ij}^{k,s,p} \leq  r^{k,s,p}.                                   
\end{aligned}
\end{equation}
However, the above linearization \eqref{bilinearcons} generally leads to a weak \LP relaxation (i.e., relaxing all binary variables to continuous variables in $[0,1]$) \cite{Chen2023}.
In  other words, the \LP relaxation of the linearized version of \eqref{minlp} could be much weaker than the (natural) continuous relaxation of \eqref{minlp}.
It is inefficient to  employ a standard solver to solve the linearized version of  \eqref{minlp}, as demonstrated in \rev{\cref{timeF} of \cref{sect:numres}}.
In the next section, we shall transform the \MINLP formulation into a \emph{novel} \MILP formulation that enjoys a much stronger polynomial-time solvable \LP relaxation than the LP relaxation of the linearized version of \eqref{minlp}.

\rev{Second, in formulation \eqref{minlp}, it is assumed that no service has been deployed in the underlying network. However, the proposed formulation \eqref{minlp} can be easily adapted to \NS problems in the online context where some services have been already deployed and new coming services need to be mapped into the network.
There are two possible approaches to achieving this.
In the first approach, we can compute the link and node capacity consumption by existing services and solve problem \eqref{minlp} with remaining link and node capacities.
In the second approach, we can solve problem \eqref{minlp}  with existing services and new coming ones, which may re-map the existing services, as to minimize the resource consumption.}

\rev{Third, we} would like to highlight the difference between our proposed formulation \eqref{minlp} and the two most closely related works \cite{Vizarreta2017} and \cite{Chen2023}.
More specifically, \rev{if we set $P=1$, then formulation \eqref{minlp} reduces to the one in \cite{Vizarreta2017}. Different} from that in \cite{Vizarreta2017} where a single path routing strategy is used for the traffic flow of each service (between \rev{the} two cloud nodes hosting \rev{the} two adjacent functions of a service), our proposed formulation allows to transmit the traffic flow of each service on (possibly) multiple paths and hence fully exploits the flexibility of traffic routing\footnote{\rev{Indeed, formulation  \eqref{minlp} with $P=1$ can be seen as a restriction of formulation  \eqref{minlp} with $P>1$, obtained by setting variables $\{  r_{ij}^{k,s,p}\}_{p \geq 2}$ and $\{r^{k,s,p}\}_{p \geq 2}$ to zero, removing variables $\{z^{k,s,p}\}_{p \geq 2}$ and related constraints from formulation  \eqref{minlp} with $P>1$.
		Due to the larger feasible region, solving formulation  \eqref{minlp} with $P>1$ can return a better feasible solution than solving formulation  \eqref{minlp} with $P=1$.}}. 
In sharp contrast to that in \cite{Chen2023}, our proposed formulation can guarantee the E2E reliability of all services.
\section{A Novel \MILP Formulation}
\label{sect:MILP}

In this section, we transform the problem \eqref{minlp} into a novel \MILP reformulation for the \NS problem which is mathematically equivalent to \eqref{minlp} but much more computationally solvable, thereby significantly facilitating the algorithmic design. 
%
\subsection{A New \MILP Formulation}
\label{subsect:MILP}
\subsubsection{New Linear Flow Conservation Constraints}
Recall that in the previous section, we use \eqref{SFC1}, \eqref{relalambdaandx11}, and \emph{nonlinear} constraint \eqref{nonlinearcons} to ensure the flow conservation for the data rates of the $p$-th path of flow $(k,s)$.
However, the intrinsic nonlinearity in \eqref{nonlinearcons} leads to an inefficient solution of formulation \eqref{minlp}.
To overcome this difficulty, we reformulate \eqref{SFC1}--\eqref{nonlinearcons} as  
\begin{align}
& \sum_{j: (i,j) \in \mathcal{{L}}} z_{ij}^{k, s,  p}\leq 1, ~\forall~i \in \mathcal{I},~k \in \mathcal{K},~s \in \mathcal{F}^k\cup \{0\},~ p \in \mathcal{P}, \label{SFC0}\\
& r_{ij}^{k, s, p} \leq z_{ij}^{k, s,p },~\forall~(i,j) \in {\mathcal{L}}, ~k \in \mathcal{K}, ~s \in \mathcal{F}^k\cup \{0\},~p \in \mathcal{P}, \label{linearcons}\\
& \sum_{p \in \mathcal{P}}\sum_{j: (j,i) \in \mathcal{{L}}} r_{ji}^{k, s, p} - \sum_{p \in \mathcal{P}}\sum_{j: (i,j) \in \mathcal{{L}}} r_{ij}^{k, s,  p}\nonumber\\
& \qquad \qquad\qquad\quad= b_{i}^{k,s}(\boldsymbol{x}),~ \forall~k \in \mathcal{K}, ~(s,i) \in \mathcal{SI}^k, \label{SFC2}
\end{align}
and \eqref{SFC3}--\eqref{SFC5} on top of next page.
\begin{figure*}[t]
	\begin{numcases}{\sum_{j: (j,i) \in \mathcal{{L}}} r_{ji}^{k, s, p} -\sum_{j: (i,j) \in \mathcal{{L}}} r_{ij}^{k, s,  p}}
	=0,&$\forall~i\in \mathcal{I}\backslash\mathcal{V},~k \in \mathcal{K},~s\in \mathcal{F}^k\cup \{0\},~p \in \mathcal{P}$;\label{SFC3}\\[3pt]
	\leq x_{i}^{k,s+1},&$\forall~i\in \mathcal{V},~k \in \mathcal{K},~s\in \mathcal{F}^k\cup \{0\}\backslash\{\ell_k\},~p \in \mathcal{P}$;\label{SFC4}\\[3pt]
	\geq -x_{i}^{k,s},&$\forall~i\in \mathcal{V},~k \in \mathcal{K},~s\in \mathcal{F}^k,~p \in \mathcal{P}$\label{SFC5}.
	\end{numcases}
	\hrulefill
\end{figure*}
Here,
\begin{multline*}
\mathcal{SI}^k = \left \{ (0,i)  :  i \in \mathcal{V}\cup \{S^k\} \right \}  \mathsmaller{\bigcup} 
\\ \left\{ (\ell_k,i)  :  i \in \mathcal{V}\cup \{D^k\} \right\}   \mathsmaller{\bigcup}  
\\ \left\{ (s,i) :  s\in \{1,2,\ldots, \ell_k-1\},~i\in \mathcal{V} \right\}.
\end{multline*}
\noindent Below we elucidate the above constraints \rev{\eqref{SFC0}--\eqref{SFC5}} one by one.   
Constraint \eqref{SFC0} ensures that there exists at most one link leaving node $i$ for the $p$-th path of flow $(k,s)$.
Constraint \eqref{linearcons} requires that, if $r_{ij}^{k,s,p}> 0$, $z_{ij}^{k,s,p}=1$ must hold.
Constraint \eqref{SFC2} ensures the flow conservation of data rates at the source and destination of flow $(k,s)$.
More specifically, it guarantees that, for flow $(k,s)$ with different source and destination, (i) the total fraction of data rate $\lambda_{s}^k$ leaving the source node $S^k$ (i.e., $s=0$) and the node hosting function $f_s^k$ (i.e., $s\in \{1,2,\ldots, \ell_k \}$) are all equal to $1$, and (ii) the total fraction of data rate $\lambda_{s}^k$ coming into the node hosting function $f_{s+1}^k$ (i.e., $s\in \{0, 1,\ldots,\ell_k -1\}$) and the destination node $D^k$ (i.e., $s=\ell_k$) are all equal to $1$.
Finally, constraints \eqref{SFC3}--\eqref{SFC5} ensure the flow conservation of the data rate at each intermediate node of the $p$-th path of flow $(k,s)$. 
In particular, when $x_{i}^{k,s}=x_{i}^{k,s+1}=0$ for some $i\in \mathcal{V}$, cloud node $i$ is an intermediate node of flow $(k,s)$, and in this case, combining \eqref{SFC4} and \eqref{SFC5}, we have 
$$ \sum_{j: (j,i) \in \mathcal{{L}}} r_{ji}^{k, s, p} -\sum_{j: (i,j) \in \mathcal{{L}}} r_{ij}^{k, s,  p}=0. $$
\subsubsection{Valid Inequalities}
To further improve the computational efficiency of the problem formulation, here we introduce two families of valid inequalities.

First, by $r^{k,s,p}\geq 0$, \eqref{nonlinearcons}, and \eqref{zzrelcons}, we have 
\begin{equation*}
	\sum_{p \in \mathcal{P}} r_{ij}^{k,s,p} = \sum_{p \in \mathcal{P}} r^{k,s,p} z_{ij}^{k,s,p} \leq z_{ij}^k\sum_{p \in \mathcal{P}} r^{k,s,p} 
\end{equation*}
which, together with \eqref{relalambdaandx11}, implies that the following inequalities 
\begin{multline}
\sum_{p \in \mathcal{P}} r_{ij}^{k,s,p} \leq z_{ij}^k, ~
\forall~(i,j)\in \mathcal{L},~k \in \mathcal{K}, ~s \in \mathcal{F}^k \cup \{0\}\label{validineq1}
\end{multline}
are valid for formulation \eqref{minlp} in the sense that \eqref{validineq1} holds at all feasible solutions $(\x, \y,\r,\z, \btheta)$ of \eqref{minlp}.

Second, combining $r^{k,s,p} \geq 0$, \eqref{nonlinearcons}, and \eqref{consdelay2funs1}, we have
\begin{equation*}
	r^{k,s,p}\theta^{k,s} \geq  	r^{k,s,p}\sum_{(i,j) \in \mathcal{{L}}}  d_{ij}  z_{ij}^{k, s, p} =\sum_{(i,j) \in \mathcal{{L}}}  d_{ij}  r_{ij}^{k, s, p}.
\end{equation*}
Summing the above inequalities for all $p \in \CP$ and using \eqref{relalambdaandx11}, we obtain the second family of valid inequalities for formulation \eqref{minlp}:
\begin{multline}
\theta^{k,s} \geq  \sum_{(i,j) \in \mathcal{L}} d_{ij} \sum_{p\in \mathcal{P}}r_{ij}^{k,s,p},~
\forall~k \in \mathcal{K}, ~s \in \mathcal{F}^k \cup \{0\}. \label{validineq2}
\end{multline}
\subsubsection{New MILP Formulation}
We are now ready to present the new formulation for the \NS problem:
\begin{align}
	\min_{\substack{\boldsymbol{x},\,\boldsymbol{y},\\\boldsymbol{r},\,\boldsymbol{z},\,\boldsymbol{\theta}}}~&  \sum_{v \in \mathcal{V}}y_v + \sigma \sum_{(i,j)\in \mathcal{L}} \sum_{k\in \mathcal{K}} \sum_{s\in \mathcal{F}^k\cup \{0\}} \sum_{p\in \mathcal{P}} \lambda_s^k r_{ij}^{k,s,p} \nonumber \\
	~~{\text{s.\,t.~~}}   & \eqref{onlyonenode}\text{--}\eqref{nodecapcons},\rev{~\eqref{linkcapcons1}\text{--}\eqref{E2Ereliability2},~\eqref{delayconstraint}\text{--}\eqref{rzbounds}, ~\eqref{consdelay2funs1}\text{--}\eqref{validineq2}}. 
	\label{milp}
	\tag{\text{MILP}}
\end{align}
Note that in the above new formulation, we do not need variables $\{r^{k,s,p}\}$. 
Moreover, the new formulation is an \MILP problem as all constraints  are \emph{linear}, which is in sharp contrast to the nonlinear constraint \eqref{nonlinearcons} in \eqref{minlp}.

\subsection{Theoretical Analysis}
\label{subsect:theo}

In this subsection, we show that formulations \eqref{minlp} and \eqref{milp} are indeed equivalent in the sense that they either are infeasible or share the same optimal solution. 
Moreover, we show that their continuous relaxations are also equivalent.

\begin{theorem}
	\label{eqformulations}
	Either problems \eqref{minlp} and \eqref{milp} are infeasible, or they are equivalent in the sense of sharing the same optimal solution.
\end{theorem}
\begin{proof}
	\ifthenelse{\longpaper = 1}{The proof can be found in Appendix \ref{appendix1}.}{The proof can be found in Appendix A of the full version of the paper \cite{Chen2024}.}
\end{proof}

\cref{eqformulations} implies that to find an optimal solution to the \NS problem, we can solve either formulation \eqref{minlp} or formulation \eqref{milp}.
Next, we further compare the solution efficiency of solving the two formulations by comparing their continuous relaxations:
\begin{align}
	\min_{\substack{\x,\,\y,\\\r,\,\z,\,\btheta}}~   & \sum_{v \in \mathcal{V}}y_v + \sigma \sum_{(i,j)\in \mathcal{L}} \sum_{k\in \mathcal{K}} \sum_{s\in \mathcal{F}^k\cup \{0\}} \sum_{p\in \mathcal{P}} \lambda_s^k r_{ij}^{k,s,p} \nonumber \\
	~{\text{s.\,t.~~}}                                                                    &   \rev{\eqref{onlyonenode}\text{--}\eqref{E2Ereliability2},~\eqref{delayconstraint}},~\eqref{rbounds},~\rev{\eqref{consdelay2funs1}},~\eqref{xyboundsR}\text{--}\eqref{rzboundsR},
	\label{nlp}
	\tag{\text{NLP}}
\end{align}
and 
\begin{align}
	\min_{\substack{\x,\,\y,\\\r,\,\z,\,\btheta}}~&  \sum_{v \in \mathcal{V}}y_v + \sigma \sum_{(i,j)\in \mathcal{L}} \sum_{k\in \mathcal{K}} \sum_{s\in \mathcal{F}^k\cup \{0\}} \sum_{p\in \mathcal{P}} \lambda_s^k r_{ij}^{k,s,p} \nonumber \\
	~~{\text{s.\,t.~~}}   & \eqref{onlyonenode}\text{--}\eqref{nodecapcons},~
	\rev{\eqref{linkcapcons1}\text{--}\eqref{E2Ereliability2},~\eqref{delayconstraint}, ~\eqref{consdelay2funs1}\text{--}\eqref{rzboundsR}},
	\label{lp}
	\tag{\text{LP-I}}
\end{align}
where
\begin{align}
	& x_{v}^{k,s},~x_v^k,~y_v\in[0,1],~\forall~v\in\mathcal{{V}}, ~k\in\mathcal{K}, ~s\in \mathcal{F}^k,\label{xyboundsR}\\
	& r_{ij}^{k, s, p}\geq 0,~z_{ij}^{k, s, p},~z_{ij}^k\in [0,1],~\theta^{k,s} \geq 0, \nonumber \\
	& \qquad \qquad ~\forall~(i,j)\in \mathcal{L},~k\in \mathcal{K},~s\in \mathcal{F}^k\cup \{0\},~p \in \mathcal{P}. \label{rzboundsR}
\end{align}
The second key result in this subsection is the equivalence between \eqref{nlp} and \eqref{lp}. 
Showing this equivalence turns out to be a highly nontrivial task. 
In particular, the proof contains two main steps: the first step is to show that problem \eqref{nlp} is equivalent to the following \LP problem \eqref{lp2}:
\begin{equation}
	\label{lp2}
	\tag{\rm{LP-II}}
	\begin{aligned}
		\min_{\substack{\x,\,\y,\\\r,\,\z,\,\btheta}} & ~~   \sum_{v \in \mathcal{V}}y_v  + \sigma \sum_{(i,j)\in \mathcal{L}} \sum_{k\in \mathcal{K}} \sum_{s\in \mathcal{F}^k\cup \{0\}} \lambda_s^k r_{ij}^{k,s} \nonumber                       \\
		{\text{s.t.~}}                                                             &~ {\eqref{onlyonenode}\text{--}\eqref{nodecapcons}},~\eqref{mediacons2-2},~ {\eqref{linkcapcons2},~\eqref{zzrelcons2},~\eqref{E2Ereliability2}},\nonumber \\
		&~ \eqref{delayconstraint},~\eqref{consdelay2funs2},~\eqref{xyboundsR},~\eqref{NLPrzbounds},~
	\end{aligned}
\end{equation}
where
\begin{align}
	&\sum_{j: (j,i) \in \mathcal{{L}}} r_{ji}^{k, s} - \sum_{j: (i,j) \in \mathcal{{L}}} r_{ij}^{k, s}=b_{i}^{k,s}(\bx),~\nonumber\\
	& \qquad \qquad\qquad\qquad\forall~i \in \I,~k \in \mathcal{K},~s \in \mathcal{F}^k\cup \{0\}, \label{mediacons2-2}\tag{\rm{5}'}\\
	& \sum_{k \in \mathcal{K}} \sum_{s\in \mathcal{F}^k \cup \{0\}} \lambda_{s}^k r_{ij}^{k, s} \leq C_{ij}, ~  \forall~(i,j) \in \mathcal{L},	\label{linkcapcons2}
	\tag{\rm{8}'}\\
	& r_{ij}^{k,s}\leq z_{ij}^k, ~\forall~(i,j) \in \CL,~k \in \mathcal{K}, ~ \forall~s \in \mathcal{F}^k \cup \{0\},\label{zzrelcons2}\tag{\rm{9}'}\\
	& \theta^{k,s} \geq \sum_{(i,j) \in \mathcal{{L}}}  d_{ij}  r_{ij}^{k, s},~\forall~k \in \mathcal{K}, ~s \in \mathcal{F}^k \cup \{0\} \label{consdelay2funs2},\tag{\rm{16}'}\\
	& r_{ij}^{k, s},~z_{ij}^{k}\in [0,1], ~\theta^{k,s}\geq 0,\nonumber\\
	& \qquad \qquad\qquad\forall~(i,j)\in \mathcal{{L}}, ~k\in \mathcal{K},~s\in \mathcal{F}^k \cup \{0\};	\tag{26'}\label{NLPrzbounds}
\end{align}
the second step is then to show the equivalence between \eqref{lp2} and \eqref{lp}. 
It is worthwhile mentioning that \eqref{lp2} plays a crucial role not only in proving the equivalence but also in the development of the \cCG algorithm in the next section.

\begin{theorem}
	\label{eqrelaxations}
 	Either problems \eqref{nlp}, \eqref{lp}, and \eqref{lp2} are infeasible, or they are equivalent in the sense of sharing the same optimal value.
\end{theorem}
\begin{proof}
	\ifthenelse{\longpaper = 1}{The proof can be found in Appendix \ref{appendix4}.}{The proof can be found in Appendix B of the full version of the paper \cite{Chen2024}.}
\end{proof}

\cref{eqrelaxations} shows the advantage of formulation \eqref{milp} over formulation \eqref{minlp}.
More specifically, in sharp contrast to the continuous \NLP relaxation of \eqref{minlp}, the continuous \LP relaxation of  \eqref{milp} is polynomial-time solvable.
In addition, when compared with the LP relaxation of the linearization version of \eqref{minlp} (i.e., \rev{replacing \eqref{nonlinearcons} and \eqref{maxdelay1} by \eqref{bilinearcons} and \eqref{consdelay2funs1}, respectively}), the LP relaxation of \eqref{milp} offers a much stronger  bound (which is as strong as that of the NLP relaxation of \eqref{minlp}).
Therefore, solving \eqref{milp} by off-the-shelf \MILP solvers should be much more computationally efficient, compared against solving \eqref{minlp}, as a strong relaxation bound plays a crucial role in the efficiency of \MILP solvers \cite{Liu2024}. 
In Section \ref{sect:numres}, we shall present computational results to further demonstrate this.

We remark that formulation \eqref{milp} is also a crucial step towards developing efficient customized algorithms for solving the \NS problem due to the following two reasons.
First, globally solving the problem provides an important benchmark for evaluating the performance of heuristic algorithms for the same problem.
More importantly, some decomposition algorithms require solving  several small subproblems of the \NS problem (e.g., problem with a single service), and employing \eqref{milp} to solve these subproblems is significantly more computationally efficient (as compared with \eqref{minlp}), thereby significantly enhancing the overall performance of the corresponding algorithm.
In the next section, we propose an efficient decomposition algorithm based on the proposed formulation \eqref{milp} for solving large-scale \NS problems.

\section{An Efficient Column Generation Algorithm}
\label{sect:CG}

Column generation (\CG) is an algorithmic framework for solving \MILP problems with a large number of variables \cite{Wolsey2021}.
It is a two-stage decomposition-based algorithm.
In the first stage, it identifies a small subset of variables (with which the corresponding \LP relaxation of the \MILP problem is solved) by decomposing the \LP relaxation of the \MILP problem into the so-called master problem and multiple subproblems,  and iteratively solving the master problem and subproblems until an optimal solution of the \LP relaxation is found.
In the second stage, it finds a high-quality solution of the original problem by  solving a restricted \MILP problem defined over the variables found in the first stage. 
In this section, we develop a \cCG algorithm for solving the \NS problem by carefully exploiting the problem's special structure. 

\subsection{A Service Pattern-based Formulation}

To facilitate the development of the \cCG algorithm, we need the pattern-based formulation  for the \NS problem.
\begin{definition}[Pattern and Pattern Set of Service $k$]
	 A pattern of service $k$ refers to a pair of the placement strategy of all functions $\left\{f_s^k\right\}_{s\in \F^k}$ and the traffic routing strategy of all flows $\{(k,s)\}_{s \in \F^k\cup \{0\}}$ in the physical network that satisfies all {E2E} QoS constraints associated with the corresponding service.  
	 The set of all patterns of service $k$ is called the pattern set of the corresponding service, denoted as $\C^k$.
\end{definition}
For each pattern $c \in \C^k$ of service $k$, let $\rev{\chi_{v}^{k,c}}\in \{0,1\}$ denote whether cloud node $v$ hosts some function(s) in the SFC of service $k$, $R^{k,c}_v$  denote the total data rate of the functions in the SFC of service $k$ processed at node $v$, 
and $R^{k,c}_{ij}$ denote the total data rate of flow $(k,s)$, $s \in \F^k\cup \{0\}$, on link $(i,j)$. 
Let $(\x, \y, \r, \z, \btheta)$ be a vector that satisfies constraints \eqref{onlyonenode}--\eqref{nodecapcons},~\rev{\eqref{linkcapcons1}\text{--}\eqref{E2Ereliability2},~\eqref{delayconstraint}\text{--}\eqref{rzbounds}, and \eqref{consdelay2funs1}--\eqref{validineq2}}, and $(\x^k, \y^k,  \r^k, \z^k, \btheta^k)$ be a subvector restricted to service $k$. 
Note that each pattern $c \in \C^k$, {that satisfies the capacity constraints of all cloud nodes and links and the E2E QoS constraints of service $k$}, corresponds to a solution $(\x^k, \y^k,  \r^k, \z^k, \btheta^k)$ and 
$\C^k$ is the set of all such feasible solutions:
\begin{equation}\label{Ckdef}
	\begin{aligned}
	& \C^k :=\{ (\bm{x}^k, \bm{y}^k, \bm{r}^k, \bm{z}^k, \bm{\theta}^k) \mid\eqref{onlyonenode}\text{--}\eqref{nodecapcons},\\
	& \qquad\qquad\qquad\rev{\eqref{linkcapcons1}\text{--}\eqref{E2Ereliability2},~\eqref{delayconstraint}\text{--}\eqref{rzbounds},~  \eqref{consdelay2funs1}\text{--}\eqref{validineq2}} \}.
	\end{aligned}
\end{equation}
By definition, the related parameters for the pattern corresponding to $(\bm{x}^k, \bm{y}^k, \bm{r}^k, \bm{z}^k, \bm{\theta}^k)$ can be computed by  
\begin{align}
	 \rev{\chi_{v}^{k,c}} & = x_{v}^{k},~\forall~v \in \V,\label{eq13}\\
	 R_v^{k,c} & = \sum_{s \in \F^k} \lambda_s^k x_{v}^{k,s}, ~\forall~v\in \V,\label{eq14}\\
	 R_{ij}^{k,c}& = \sum_{s \in \F^k\cup \{0\}} \lambda_s^k \sum_{p \in \CP} r_{ij}^{k,s,p},~\forall~(i,j) \in \CL.\label{eq12}
\end{align}
Once $(\bm{x}^k, \bm{y}^k, \bm{r}^k, \bm{z}^k, \bm{\theta}^k)$ is given, the parameters given in \eqref{eq13}--\eqref{eq12} are constants.


We use variable $t^{k,c}\in \{0,1\}$ to denote whether or not pattern $c$ is used for service $k$.
By definition, each service $k$ must choose one pattern  in $\C^k$:
\begin{equation}
	\sum_{c \in \mathcal{C}^k}t^{k,c} =1, ~ \forall~k \in \mathcal{K}. \qquad \label{atmostone}        
\end{equation}
If pattern $c$ is chosen and cloud node $v$ hosts some function(s) of service $k$, then cloud node $v$ must be activated:
\begin{equation}
	 \sum_{c \in  \C^k}\rev{\chi_{v}^{k,c}}t^{k,c}  \leq y_v,~\forall~v \in \mathcal{V}, ~ k \in \mathcal{K}, \label{cover}
\end{equation}
where we recall that $y_v \in \{0,1\}$ denotes whether or not cloud node $v$ is activated.  
To enforce the node and link capacity constraints, we need 
\begin{align}
	& \sum_{k \in \mathcal{K}} \sum_{c \in  \C^k} R_{v}^{k,c}  t^{k,c}   \leq \mu_v y_v ,  ~\forall~v \in \mathcal{V} ,~
	\label{decision}\\
	& \sum_{k \in \mathcal{K}} \sum_{c \in  \C^k} R^{k,c}_{ij}t^{k,c} \leq C_{ij} , ~ \forall ~(i,j) \in \mathcal{L}.~\quad   \label{resources} 
\end{align}

We are ready to present the service pattern-based formulation for the \NS problem:
\begin{equation}	\label{flow_form} \tag{P-MILP}
	\begin{aligned}
		&\min_{\boldsymbol{y},\,\boldsymbol{t}}
		& &   \sum_{v \in \mathcal{V}} y_v + \sigma  \sum_{k \in \mathcal{K}} \sum_{c \in \C^k} \left(\sum_{(i,j)\in \CL} R^{k,c}_{ij}\right)t^{k,c}              \\
		&~~ \text{s.t.}              &   &     \eqref{atmostone}\text{--}\eqref{resources}, \\
		&&&	y_v \in \{0,1\}^n,~\forall~v \in \V,\\
		& &&  t^{k,c} \in \{0,1\}, ~ \forall~k\in \mathcal{K},~c \in  \C^k.
	\end{aligned}
\end{equation}

\subsection{Proposed Algorithm}

Although \eqref{flow_form} is an \MILP problem, its huge number of  variables $\{t^{k,c}\}$ (as the number of feasible solutions in $\C^k$ could even be infinite) makes it unrealistic to be directly solved.
In order to avoid solving problem \eqref{flow_form} with a huge number of patterns, we propose a \cCG algorithm for solving problem \eqref{flow_form}. 
The \cCG algorithm is a two-stage algorithm where the  first step  applies a \CG procedure to identify an effective but usually small subset of patterns $\tilde{\C}^k \subseteq {\C}^k$ for each $k \in \K$, and the second stage solves  problem \eqref{flow_form} with $\C^k$ replaced by $\tilde{\C}^k$ for all $k \in \K$ to construct a high-quality solution for problem \eqref{flow_form}. 


\subsubsection{Algorithmic Framework}
Let us first consider the \LP relaxation of problem \eqref{flow_form}:
\begin{equation}	\label{pLP} \tag{P-LP}
	\begin{aligned}
		&\min_{\boldsymbol{y},\,\boldsymbol{t}}
		& &   \sum_{v \in \mathcal{V}} y_v + \sigma  \sum_{k \in \mathcal{K}} \sum_{c \in \C^k} \left(\sum_{(i,j)\in \CL} R^{k,c}_{ij}\right)t^{k,c}              \\
		&~~ \text{s.t.}              &   &     \eqref{atmostone}\text{--}\eqref{resources}, \\
		&&&	y_v \in [0,1]^n,~\forall~v \in \V,\\
		& &&  t^{k,c} \geq 0, ~ \forall~k\in \mathcal{K},~c \in  \C^k.
	\end{aligned}
\end{equation}
Although problem \eqref{pLP} still has a huge number of patterns (or variables $\{t^{k,c}\}$), it can be solved by the \CG algorithm in which only a small subset of patterns is initially considered, and additional patterns are gradually added (when needed) until an optimal solution of the original problem is found. 

Let $\tilde{\C}^k$ be a subset of $\C^k$, $k \in \K$, and consider the following \LP problem: 
\begin{equation}
	\label{flow_lp}\tag{RLP}
\begin{aligned}
	&\min_{\boldsymbol{y},\,\boldsymbol{t}}
	& &   \sum_{v \in \mathcal{V}} y_v + \sigma  \sum_{k \in \mathcal{K}} \sum_{c \in \tilde{\C}_k} \left(\sum_{(i,j)\in \CL} R^{k,c}_{ij}\right)t^{k,c}       \\
	&~~ \text{s.t.}              &   &        {\eqref{atmostone}\text{--}\eqref{resources}},\\
	&&& y_v \leq 1, ~\forall~v \in \V,  \\
	& &&  t^{k,c} \geq0, ~ \forall~k\in \mathcal{K},~c \in  \tilde{\C}^k,\\
	& & &y_v\geq 0, ~\forall~v \in \V.
\end{aligned}
\end{equation}
{Problem \eqref{flow_lp} is a restriction of problem \eqref{pLP} (called \emph{master problem}) obtained by setting $t^{k,c}=0$ for all $k \in \K$ and $c \in \C^k \backslash \tilde{\C}^k$ in probelm \eqref{pLP}.}
Let $(\boldsymbol{\alpha},\boldsymbol{\beta},\boldsymbol{\pi},\boldsymbol{\eta},\boldsymbol{\zeta})$ be the dual variables corresponding to constraints \eqref{atmostone}, \eqref{cover}, \eqref{decision}, \eqref{resources}, and $y_v \leq 1$ for $v \in \V$, respectively.
Then we can get the dual of problem \eqref{flow_lp} as follows:
\begin{equation}
	\label{dlp}\tag{D-RLP}
	\begin{aligned}
		&\max_{\substack{\boldsymbol{\alpha},\,\boldsymbol{\pi},\, \boldsymbol{\eta},\\\, \boldsymbol{\beta},\, \boldsymbol{\zeta}}}
		& &   \sum_{k \in \K} \alpha_k + \sum_{(i,j)\in \CL} C_{ij}\beta_{ij} +\sum_{v \in \V}\zeta_v      \\
		&~~~ \text{s.t.}              &   &        \alpha_k+\sum_{v \in \mathcal{V}} \rev{\chi_{v}^{k,c}} \pi^{k}_v  +  \sum_{v \in \mathcal{V}} R_v^{k,c}\eta_v +\sum_{(i,j)\in \mathcal{L}}  R^{k,c}_{ij}\beta_{ij}, \\
		& &&\qquad \qquad \leq \sigma \sum_{(i,j)\in \CL} R_{ij}^{k,c},~ ~\forall~k \in \K,~c \in \tilde{\C}^k,\\
		&                          &       &   -\sum_{k \in \K} \pi_v^k - \mu_v \eta_v + \zeta_v \leq 1,~\forall~ v\in \V,\\
		& & & \pi_v^k, ~\eta_v,~\zeta_v \leq 0,~\forall~v \in \V,~k \in \K, \\
		& & & \beta_{ij}\leq  0, ~\forall~ (i,j)\in \CL. 
	\end{aligned}
\end{equation}
We can obtain a primal solution $(\bar{\bm{y}}, \bar{\bm{t}})$ (if the problem is feasible) and a dual solution $(\bar{\bm{\alpha}},\bar{\bm{\beta}},\bar{\bm{\pi}},\bar{\bm{\eta}},\bar{\bm{\zeta}})$ by employing the primal or dual simplex algorithm to solve problem \eqref{flow_lp} \cite[Chapter 3]{Dantzig1997}.
As the all-zero vector is feasible, one of the following two cases must happen: 
(i) $(\bar{\bm{\alpha}},\bar{\bm{\beta}},\bar{\bm{\pi}},\bar{\bm{\eta}},\bar{\bm{\zeta}})$ is an optimal solution of problem \eqref{dlp} and problem \eqref{flow_lp} has an optimal solution  $(\bar{\bm{y}}, \bar{\bm{t}})$;
(ii) $(\bar{\bm{\alpha}}, \bar{\bm{\beta}}, \bar{\bm{\pi}}, \bar{\bm{\eta}}, \bar{\bm{\zeta}})$  is an unbounded ray of problem \eqref{dlp} for which $$\sum_{k \in \K} \bar{\alpha}_k + \sum_{(i,j)\in \CL} C_{ij} \bar{\beta}_{ij} +\sum_{v \in \V}\bar{\zeta}_v >0,$$ and problem \eqref{flow_lp} is infeasible.
Notice that in both cases, the primal or dual simplex algorithm will return this dual solution $(\bar{\bm{\alpha}},\bar{\bm{\beta}},\bar{\bm{\pi}},\bar{\bm{\eta}},\bar{\bm{\zeta}})$; see \cite[Chapter 3]{Dantzig1997}.

Leveraging the dual solution $(\bar{\bm{\alpha}},\bar{\bm{\beta}},\bar{\bm{\pi}},\bar{\bm{\eta}},\bar{\bm{\zeta}})$, we can determine whether or not the restriction \eqref{flow_lp} is equivalent to the original problem \eqref{pLP}. 
Specifically, (i) if
\begin{equation}\label{optcond1}
	\begin{aligned}
	& \bar{\alpha}_k+\sum_{v \in \mathcal{V}} \bar{\pi}^{k}_v\rev{\chi_{v}^{k,c}}  +  \sum_{v \in \mathcal{V}} \bar{\eta}_v R_v^{k,c} \\
	& \qquad \qquad\qquad\qquad +\sum_{(i,j)\in \mathcal{L}} \bar{\beta}_{ij} R^{k,c}_{ij} \leq \sigma \sum_{(i,j)\in \CL} R_{ij}^{k,c} 
	\end{aligned}
\end{equation}
holds for all $c \in  \C^k$ (or equivalently, for all $c \in \C^k \backslash \tilde{\C}^k$) and $k \in \K$, then $(\bar{\bm{\alpha}},\bar{\bm{\beta}},\bar{\bm{\pi}},\bar{\bm{\eta}},\bar{\bm{\zeta}})$ is also a dual solution of problem \eqref{pLP}.
Therefore, problems \eqref{flow_lp} and \eqref{pLP} are equivalent and $\tilde{\C}^k$ for all $k \in \K$ are the desired subsets.
(ii) Otherwise, there must exist some $k \in \K$ and $c \in \C^k$ for which \eqref{optcond1} is violated. 
In this case, we can add these patterns into subsets $\tilde{\C}^k$ to obtain a new enhanced \LP problem. 
We then solve problem \eqref{pLP} again to obtain new primal and dual  solutions. 
This process is repeated until case (i) happens.
The overall \cCG algorithm for solving the \NS problem
is summarized in \cref{alg1}.

\begin{algorithm}[!t]
	\caption{The \cCG algorithm for solving the \NS problem}
	\label{alg1}
	\begin{algorithmic}[1]
		\renewcommand{\algorithmicrequire}{\textbf{Input:}}
		\renewcommand{\algorithmicensure}{\textbf{Output:}}
		\renewcommand{\algorithmiccomment}[1]{//\,\texttt{#1}}
		\STATE Set $\tau \leftarrow 0$ and $\text{IterMax}\geq 1$;
		\\[5pt]
		{\COMMENT{Stage 1: solving the \LP problem \eqref{pLP}}}\\[5pt]
		\FOR{$k \in \K$}
		\STATE Solve the subproblem of \eqref{milp} with $\K$ being $\{k\}$; 
		\STATE If the subproblem has a feasible solution $({\x}^k, {\y}^k, {\r}^k,{\z}^k, {\btheta}^k)$, initialize $\tilde{\C}^k\leftarrow\{c\}$, where $c$ denotes the corresponding pattern; otherwise, {\bf stop} and claim that the \NS problem is infeasible; 
		\ENDFOR
		\WHILE {$\tau \leq\text{IterMax}$} 
		\STATE Solve the LP problem \eqref{flow_lp} to obtain a dual solution $(\bar{\boldsymbol{\alpha}},\bar{\boldsymbol{\beta}},\bar{\boldsymbol{\pi}},\bar{\boldsymbol{\eta}},\bar{\boldsymbol{\zeta}})$;
		\FOR{$k \in \K$}
		\STATE Call \cref{alg0} to find a new pattern or declare that no new pattern exists for service $k$. 
		If a new pattern $c$ is found, set $\tilde{\C}^k \leftarrow \tilde{\C}^k\cup \{c\}$;
		\ENDFOR
		\IF{no new pattern is found for any service $k \in \K$}
		\IF{$(\bar{\boldsymbol{\alpha}},\bar{\boldsymbol{\beta}},\bar{\boldsymbol{\pi}},\bar{\boldsymbol{\eta}},\bar{\boldsymbol{\zeta}})$ is an unbounded ray}
		\STATE {\bf Stop} and claim that problems \eqref{pLP} and \eqref{flow_form} are infeasible;
		\ELSE
		\STATE Go to the second stage in step 20;
		\ENDIF
		\ENDIF
		\STATE Set $\tau \leftarrow \tau+1$;
		\ENDWHILE
		\\[5pt]
		{\COMMENT{Stage 2: solving the restriction problem of \eqref{flow_form}}}	\\[5pt]
		\STATE Solve the restriction of problem \eqref{flow_form} with $\C^k$ replaced by $\tilde{\C}^k$ for all $k \in \K$ in  \eqref{flow_form}. 
	\end{algorithmic} 
\end{algorithm}

\subsubsection{Finding a New Pattern}
In this part, we attempt to find a new pattern $c \in \C^k$ for which  \eqref{optcond1} does not hold (if such a pattern exists); see step 9 of \cref{alg1}.
Recall that a pattern $c$ for service $k\in\K$ is a feasible solution $(\bm{x}^k, \bm{y}^k, \bm{r}^k, \bm{z}^k, \bm{\theta}^k)$ of $\C^k$ defined in \eqref{Ckdef}.
Therefore, for each $k \in \K$, to determine whether \eqref{optcond1} holds for all $c \in \C^k$, one can solve problem 
\begin{align}
	&\max_{\substack{\boldsymbol{x}^k,\,\boldsymbol{y}^k,\\\boldsymbol{r}^k,\,\boldsymbol{z}^k,\,\boldsymbol{\theta}^k}}
	&                          &                    \bar{\alpha}_k+ \sum_{v \in \mathcal{V}} (\bar{\pi}^{k}_v \rev{\chi_{v}^{k,c}}+ \bar{\eta}_v  R_v^{k,c})+ \sum_{(i,j)\in \mathcal{L}} (\bar{\beta}_{ij} -\sigma)R^{k,c}_{ij}\nonumber \\
	& ~~~~{\text{s.t.~}}                                                                    &  & 	
	\eqref{onlyonenode}\text{--}\eqref{nodecapcons},~\rev{\eqref{linkcapcons1}\text{--}\eqref{E2Ereliability2},~\eqref{delayconstraint}\text{--}\eqref{rzbounds}, ~ \eqref{consdelay2funs1}\text{--}\eqref{validineq2}}
	\tag{$\text{SP}^k$}\label{subP}
\end{align}
where $\rev{\chi_{v}^{k,c}}$, $R_{v}^{k,c}$, and $R_{ij}^{k,c}$ are defined in \eqref{eq13}--\eqref{eq12}, respectively.
If the optimal value $\nu^k$ of problem \eqref{subP} is less than or equal to zero, then \eqref{optcond1} holds for all $c \in \C^k$; otherwise, \eqref{optcond1} is violated by the pattern corresponding to the optimal solution $(\bm{x}^k, \bm{y}^k, \bm{r}^k, \bm{z}^k, \bm{\theta}^k)$ of problem \eqref{subP}.
\rev{Problem \eqref{subP} is an \MILP problem that can be solved to global optimality using the branch-and-bound algorithm of state-of-the-art \MILP solvers like Gurobi and CPLEX.}

To further reduce the computational cost of solving the \MILP problem \eqref{subP}, we can first solve its \LP relaxation with the optimal value $\nu_\LP^k$:
if $\nu_\LP^k \leq 0$, then by $\nu^k \leq \nu_\LP^k$, $\nu^k \leq 0$ must hold, and thus \eqref{optcond1} holds for all $c \in \C^k$;
otherwise, we solve the \MILP problem \eqref{subP} to check whether there exists a new pattern $c \in \C^k$ by which \eqref{optcond1} is violated.
In particular, using the same techniques in deriving problem \eqref{lp2}, we can show that the \LP relaxation of problem \eqref{subP} is equivalent to
\begin{align}	\label{splp2}
	\tag{${\text{SP}}_{\LP}^k$}
	\begin{aligned}
	\max_{\substack{\boldsymbol{x}^k,\,\boldsymbol{y}^k,\\\boldsymbol{r}^k,\,\boldsymbol{z}^k,\,\boldsymbol{\theta}^k}} &    \alpha_k+ \sum_{(i,j)\in \mathcal{L}} (\beta_{ij} -\sigma)\sum_{s \in \F^k\cup \{0\}} \lambda_s^k r_{ij}^{k,s}\nonumber\\
	& \qquad \qquad+ \sum_{v \in \mathcal{V}} \left(\pi^{k}_v  x_v^{k}+ \eta_v  \sum_{s \in \F^k} \lambda_s^k x_{v}^{k,s}
	\right) \nonumber                       \\
	{\text{s.t.~~~}}                                                             & {\eqref{onlyonenode}\text{--}\eqref{nodecapcons}},~\eqref{mediacons2-2},~ {\eqref{linkcapcons2},~\eqref{zzrelcons2},~\eqref{E2Ereliability2},} \nonumber \\
	&                                                                                \eqref{delayconstraint},~\eqref{consdelay2funs2},~\eqref{xyboundsR},~\eqref{NLPrzbounds}.
	\end{aligned}
\end{align}
\rev{Problem \eqref{splp2} is an \LP problem that can be solved to global optimality using the simplex or interior-point algorithm of off-the-shelf \LP solvers like Gurobi and CPLEX.}

\begin{algorithm}[!t]
	\caption{The algorithm for finding a new pattern for service $k$}
	\label{alg0}
	\begin{algorithmic}[1]
		\renewcommand{\algorithmicrequire}{\textbf{Input:}}
		\renewcommand{\algorithmicensure}{\textbf{Output:}}
		\renewcommand{\algorithmiccomment}[1]{//\,\texttt{#1}}
		\STATE Solve problem \eqref{splp2} to obtain an optimal solution $(\bar{\x}^k, \bar{\y}^k, \bar{\r}^k, \bar{\z}^k, \bar{\btheta}^k)$ and the optimal value ${\nu}_\LP^k$;
		\IF{$\nu_\LP^k > 0$}
		\ifthenelse{\longpaper = 1}{\STATE Use \eqref{tmeq6-0}, \eqref{tmeq6-1}, and \eqref{tmeq8-1} to recover an optimal solution $({\x}^k, {\y}^k, {\r}^k,{\z}^k, {\btheta}^k)$ of  the (natural) \LP relaxation of problem \eqref{subP};}{\STATE Use Eqs. \rev{(45), (46), and (48)} in the full version of the paper \cite{Chen2024} to recover an optimal solution $({\x}^k, {\y}^k, {\r}^k,{\z}^k, {\btheta}^k)$ of  the (natural) \LP relaxation of problem \eqref{subP};}
		\IF{all constraints in problem \eqref{subP} are satified at $({\x}^k, {\y}^k, {\r}^k,{\z}^k, {\btheta}^k)$}
		\STATE \textbf{Stop} and return the corresponding new pattern $c$ for service $k$;
		\ELSE
		\STATE Solve subproblem \eqref{subP} to obtain a solution $({\x}^k, {\y}^k, {\r}^k,{\z}^k, {\btheta}^k)$ and the optimal value $\nu^k$;
		\IF{$\nu^k > 0$}
		\STATE \textbf{Stop} and return the corresponding new pattern $c$ for service $k$;
		\ELSE
			\STATE \textbf{Stop} and claim that no new pattern for service $k$ exists;
		\ENDIF
		\ENDIF
		\ENDIF
	\end{algorithmic} 
\end{algorithm}

Two important remarks on problem \eqref{splp2} are in order. 
First, compared with the (natural) \LP relaxation of problem \eqref{subP}, problem \eqref{splp2} has a much smaller number of variables and constraints, and thus is more computationally solvable.  
\ifthenelse{\longpaper = 1}{Second, for an optimal solution $(\bar{\x}^k, \bar{\y}^k, \bar{\r}^k, \bar{\z}^k, \bar{\btheta}^k)$ of problem \eqref{splp2}, we can recover an optimal solution $({\x}^k, {\y}^k, {\r}^k,{\z}^k, {\btheta}^k)$  of  the (natural) \LP relaxation of problem \eqref{subP} using \eqref{tmeq6-0}, \eqref{tmeq6-1}, and \eqref{tmeq8-1}; see  \cref{lpproperty2} of Appendix \ref{appendix4}.}{Second, for an optimal solution $(\bar{\x}^k, \bar{\y}^k, \bar{\r}^k, \bar{\z}^k, \bar{\btheta}^k)$ of problem \eqref{splp2}, we can recover an optimal solution $({\x}^k, {\y}^k, {\r}^k,{\z}^k, {\btheta}^k)$  of  the (natural) \LP relaxation of problem \eqref{subP} using Eqs. \rev{(45), (46), and (48)} in the full version of the paper \cite{Chen2024}.}
If $({\x}^k, {\y}^k, {\r}^k,{\z}^k, {\btheta}^k)$ is integral (i.e., it satisfies all constraints in \eqref{subP}), then it must also be an optimal solution of problem \eqref{subP}.
In this case, we do not need to solve MILP problem \eqref{subP}, which saves a lot of computational efforts.

The algorithmic details on finding a new pattern $c$ for service $k$ is summarized in \cref{alg0}.

\subsection{Analysis Results and Remarks}
\label{subsec:remark}

In this subsection, we provide some analysis results and remarks on the proposed \cCG algorithm. 

First, although the number of feasible solutions in $\C^k$ could be infinite, to solve \eqref{pLP} in the first stage, only the extreme points of the convex hull of $\C^k$,
denoted as $\conv(\C^k)$, are required \cite[Chapter 11]{Wolsey2021}\footnote{\noindent The point $(\x^k,\y^k,\r^k, \z^k, \btheta^k)$ is an extreme point of $\conv(\C^k)$ if it is not a proper convex combination of two distinct points in $\conv(\C^k)$.}. 
The number of extreme points of a polyhedron such as $\conv(\C^k)$ is finite \cite[Theorem 2.2]{Wolsey2021}. 
From this and the fact that solving problem \eqref{subP} by the \LP-based branch-and-bound algorithm (with the simplex method being used to solve the corresponding LPs) can return an extreme point of $\conv(\C^k)$ or declare the \rev{infeasibility} of \eqref{subP}, we know that steps 6--19 of \cref{alg1} with a sufficiently large $\text{IterMax}$ will either {identify} an optimal solution of problem \eqref{pLP} (if the problem is feasible) or declare its {infeasibility} (if not).
Although the worst-case iteration complexity of  steps 6--19 grows exponentially fast with the number of variables, our simulation results show that it usually finds an optimal solution of problem \eqref{pLP} within a small number of iterations (and thus the cardinalities of sets $\{\tilde{\C}^k\}$ are small); see \cref{sect:numres} further ahead.

Second, in the first stage of \cref{alg1}, each iteration needs to solve at most $|\mathcal{\K}|$ problems of the form \eqref{subP} and one problem of the form \eqref{flow_lp}.
These problems are much  easier to solve than the original problem \eqref{milp}.
To be specific, although problem \eqref{subP} is still an \MILP problem, {both of its numbers of variables and constraints are}  $\mathcal{O}((|\I|+|\CL|)|\CP| \ell_k)$.
This is significantly smaller than those of problem \eqref{milp},
which are $\mathcal{O}((|\I|+|\CL|)|\CP| \sum_{k \in \K}\ell_k)$, especially when the number of services is large.
Moreover, problem \eqref{flow_lp} is an LP problem, which is polynomial-time solvable. 
In particular, the worst-case complexity of solving problem \eqref{flow_lp} can be upper bounded by  \cite[Chapter 6.6.1]{Ben-Tal2001}:
$$\mathcal{O}\left({\left(\sum_{k \in \K}|\tilde{\C}^k|+|\V||\K|+ |\CL|\right)}^{3.5}\right).$$
Notice that the restriction problem of \eqref{flow_form} that needs to be solved in the second stage of \cref{alg1} is also a small-scale \MILP problem (as the cardinalities of sets $\{\tilde{\C}^k\}$ are often small), which makes it much easier to solve than the original problem \eqref{milp}.
Therefore, the proposed \cCG algorithm is particularly {suitable to solve} large-scale \NS problems.

Finally, the proposed \cCG algorithm can be interpreted as an \LP relaxation rounding algorithm, where the \LP relaxation problem \eqref{pLP} is solved and the optimal rounding strategy in \cite{Berthold2013} is applied to its solution (i.e., by solving problem \eqref{flow_form} with $\C^k$ replaced by $\tilde{\C}^k$).	
Moreover, using the results \cite[Chapter 11.2]{Wolsey2021}, we can show that the \LP relaxation problem \eqref{pLP} is equivalent to 
\begin{align}
	\min_{\substack{\x,\,\y,\\\r,\,\z,\,\btheta}}~&  \sum_{v \in \mathcal{V}}y_v + \sigma \sum_{(i,j)\in \mathcal{L}} \sum_{k\in \mathcal{K}} \sum_{s\in \mathcal{F}^k\cup \{0\}} \sum_{p\in \mathcal{P}} \lambda_s^k r_{ij}^{k,s,p} \nonumber \\
	~~{\text{s.\,t.~~}} &  \eqref{nodecapcons},~\eqref{linkcapcons1},~(\x^k, \y^k, \r^k, \z^k, \btheta^k) \in \conv(\C^k),~\forall~k \in \K.
 	\label{convR}
\end{align}
Theoretically, 
problem \eqref{convR} can provide a relaxation bound that is even stronger  than that of the \LP relaxation \eqref{lp}, which can be rewritten as  
\begin{align}
	\min_{\substack{\x,\,\y,\\\r,\,\z,\,\btheta}}~&  \sum_{v \in \mathcal{V}}y_v + \sigma \sum_{(i,j)\in \mathcal{L}} \sum_{k\in \mathcal{K}} \sum_{s\in \mathcal{F}^k\cup \{0\}} \sum_{p\in \mathcal{P}} \lambda_s^k r_{ij}^{k,s,p} \nonumber \\
	~~{\text{s.\,t.~~}} &  \eqref{nodecapcons},~\eqref{linkcapcons1},~(\x^k, \y^k, \r^k, \z^k, \btheta^k) \in \C_{\text{L}}^k,~\forall~k \in \K.
	\label{convL}
\end{align}
Here, 
\begin{equation}
	\begin{aligned}
	& \C_{\text{L}}^k :=\{ (\bm{x}^k, \bm{y}^k, \bm{r}^k, \bm{z}^k, \bm{\theta}^k) \mid \\
	& \qquad  \rev{\eqref{onlyonenode}\text{--}\eqref{nodecapcons},~\eqref{linkcapcons1}\text{--}\eqref{E2Ereliability2},~\eqref{delayconstraint}\text{--}\eqref{rzbounds},~  \eqref{consdelay2funs1}\text{--}\eqref{rzboundsR}}\}
	\end{aligned}
\end{equation}
is the linearization of $\C^k$
satisfying $ \conv(\C^k) \subseteq \C_{\text{L}}^k$.
This theoretical interpretation sheds important insights into why the proposed cCG algorithm is
 able to return a high-quality feasible solution for the \NS problem, as will be further demonstrated in our experiments in \cref{sect:numres}.

\section{Numerical Simulation}
\label{sect:numres}

In this section, we present numerical results to demonstrate the effectiveness of the proposed formulations and the efficiency of the proposed \cCG algorithm.
More specifically, we first present numerical results to compare the computational efficiency of the proposed formulations \eqref{milp} and \eqref{minlp} in \cref{subsec:MILPvsMINLP}.
Then, in \cref{subsec:MILPvsEformulations}, we perform numerical experiments to illustrate the effectiveness of the proposed formulation \eqref{milp} for the \NS problem over the state-of-the-art  formulations in \cite{Chen2023} and  \cite{Vizarreta2017}. 
Finally, in \cref{subsec:CGalg}, we present numerical results to illustrate the efficiency of the proposed \cCG algorithm for solving the \NS problem.

All experiments were {conducted} on a server with 2 Intel Xeon E5-2630 processors and 98 GB of RAM, {running} the Ubuntu GNU/Linux Server 20.04 x86-64 operating system.
\rev{In formulations \eqref{minlp}, \eqref{milp}, and \eqref{flow_form}, unless otherwise
stated, we choose $\sigma = 0.0005$ and $P = 2$. }
In \cref{alg1}, we set $\text{IterMax} = 100$.
We use CPLEX 20.1.0 to solve all \LP and \MILP problems (including the linearized version of \eqref{minlp}). 
When solving the \MILP problems, the time limit was set to 1800 seconds.

{The fish network topology in \cite{Zhang2017}, which contains 112 nodes and 440 links, including 6 cloud nodes, is employed in our experiments.
The capacities of cloud nodes and links are randomly generated within $ [50,100] $ and $ [7,77] $, respectively;
the NFV and communication delays on the cloud nodes and links are randomly selected from  $\{3,4,5,6\}$ and $\{1,2\}$, respectively; 
the reliabilities of cloud nodes and links are randomly chosen in $[0.991,0.995]$ and $[0.995,0.999]$, respectively.}
For each service $k$, node $S^k$ is randomly chosen from the available nodes and node $D^k$ is set to be the common destination node; SFC $ \F^k$ contains 3 functions randomly generated from $ \{f^1,f^2, \ldots, f^5\} $; $ \lambda^k_s $'s are the data rates which are all set to be the same integer value, randomly chosen from $ [1,11] $; $ \Theta^k $ and $\Gamma^k$ are set to $ 20+(3*\text{dist}_k+\alpha) $ and $ 0.99^2{(\text{dist}'_k)}^4 $, where $ \text{dist}_k $ and $ \text{dist}'_k $ are the shortest paths between nodes $ S^k $ and $ D^k $ (in terms of delay and reliability, respectively) and $ \alpha $ is randomly chosen in $[0,5]$.
The parameters mentioned above are carefully selected to ensure that the constraints in the \NS problem are neither too tight nor too loose.
For each fixed number of services, we randomly generate 100 problem instances and the results reported in this section are averages over the 100 instances.

%

\subsection{Comparison of Proposed Formulations \eqref{milp} and \eqref{minlp}}
\label{subsec:MILPvsMINLP}

\begin{figure}[t]
	\centering
	\includegraphics[height=\figuresize]{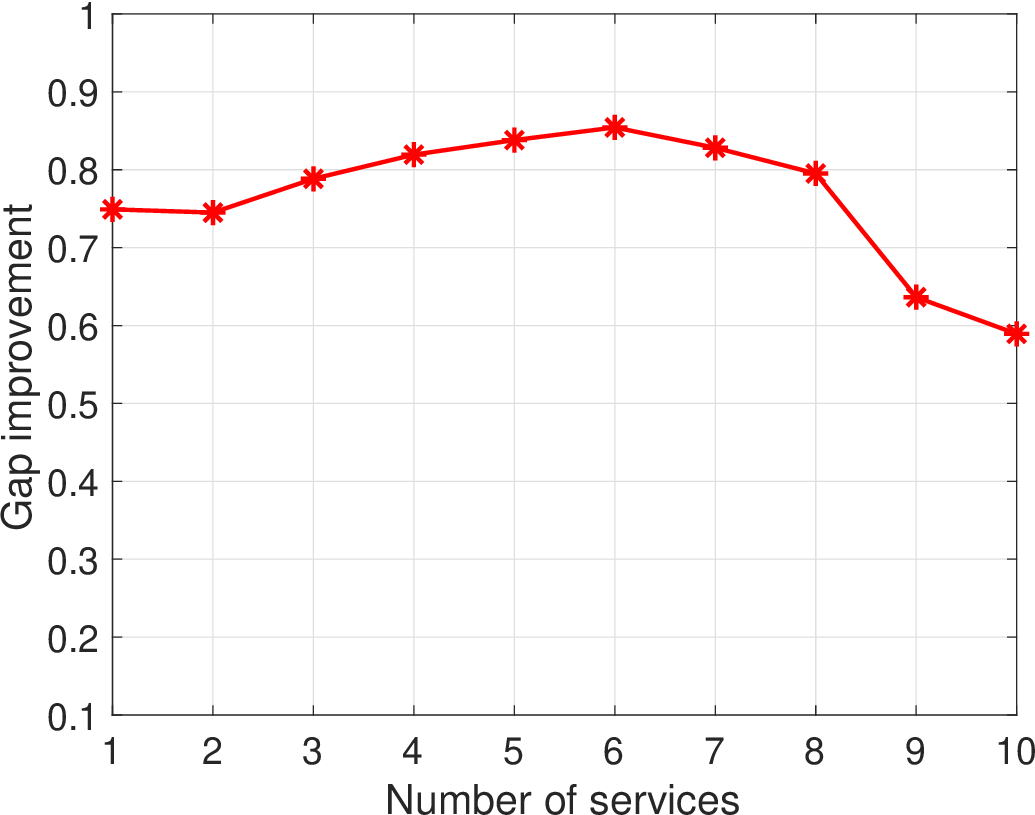}
	\caption{Average relative optimality gap improvement of the \LP relaxation of formulation \eqref{milp} over the \LP relaxation of (the linearized version of) formulation \eqref{minlp}.}
	\label{gap}
\end{figure}

\begin{figure}[t]
	\centering
	\includegraphics[height=\figuresize]{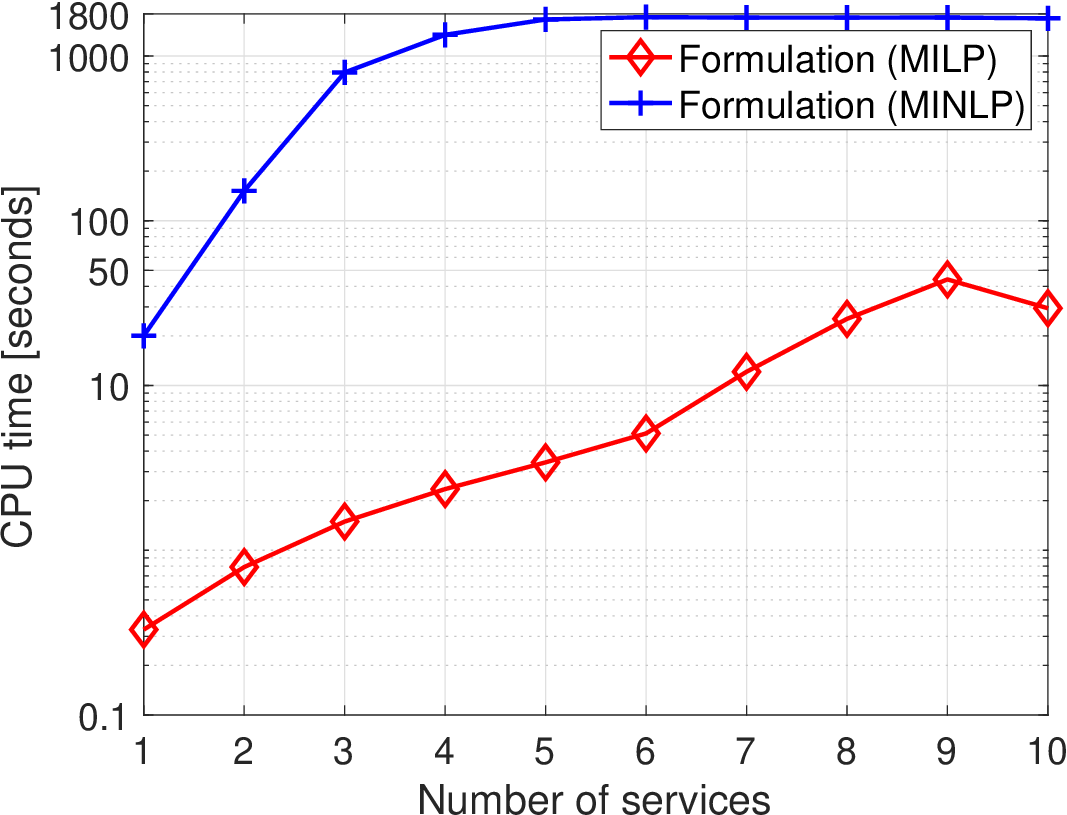}
	\caption{Average CPU time taken by solving formulations \eqref{milp} and \eqref{minlp}.}
	\label{timeF}
\end{figure}

In this subsection, we compare the performance of solving formulation \eqref{milp} and the linearized version of formulation \eqref{minlp} (i.e.,  \rev{replacing \eqref{nonlinearcons} and \eqref{maxdelay1} by \eqref{bilinearcons} and \eqref{consdelay2funs1}, respectively}).
To do this, we first compare the tightness of their continuous relaxations, i.e., problem \eqref{lp} and the linearized version of \eqref{nlp} (denoted as (NLP-L)).
We compare the relative optimality gap improvement, defined by
\begin{equation}
	\label{eq1}
	\frac{\nu(\text{LP-I})-\nu(\text{NLP-L})}{\nu(\text{MILP})-\nu(\text{NLP-L})},
\end{equation}
where $\nu(\text{LP-I})$, $\nu(\text{NLP-L})$, and $\nu(\text{MILP})$ are the optimal values of problems \eqref{lp}, (NLP-L), and \eqref{milp}, respectively.
The gap improvement in \eqref{eq1}, 
a widely used performance
measure in the integer programming community \cite{Vielma2010,Fukasawa2011},
quantifies the improved tightness of relaxation \eqref{lp} over that of relaxation (NLP-L).
The larger the gap improvement, the stronger the relaxation \eqref{lp} (as compared with relaxation (NLP-L)).
As can be observed from \cref{gap}, in all cases, the gap improvement is larger than $0.5$. 
Indeed, in most cases, the gap improvement is larger than $0.7$, which shows that the continuous relaxation of \eqref{milp} is indeed much stronger than the continuous relaxation of the linearized version of \eqref{minlp} in terms of providing a much better relaxation bound.

Next, we compare the computational efficiency of solving formulations \eqref{milp} and \eqref{minlp}.
Fig. \ref{timeF} plots the average CPU time taken by solving formulations \eqref{milp} and \eqref{minlp} versus the number of services.
From the figure, it can be clearly seen that it is much more efficient to solve \eqref{milp} than \eqref{minlp}.
In particular, in all cases, the CPU times of solving \eqref{milp} are  within $50$ seconds, while that of solving \eqref{minlp} are larger than $1000$ seconds when $|\mathcal{K}|\geq 4$.
We remark that when $|\mathcal{K}|\geq 5$, CPLEX  failed to solve \eqref{minlp}  within $1800$ seconds in almost all cases.

The above results clearly demonstrate the advantage of our new way of formulating the flow conservation constraints \eqref{SFC0}--\eqref{SFC5} for the data rates over those in \cite{Promwongsa2020,Chen2023} and the proposed valid inequalities \eqref{validineq1}--\eqref{validineq2}, i.e., it enables to provide an \MILP formulation with a much stronger \LP relaxation and thus can effectively make the \NS problem much more computationally solvable.
Due to this, in the following, we only use and discuss formulation \eqref{milp}.
%

\subsection{Comparison of Proposed Formulation \eqref{milp} and Those in  \cite{Chen2023} and \cite{Vizarreta2017}}
\label{subsec:MILPvsEformulations}

\begin{figure}[t]
	\centering
	\includegraphics[height=\figuresize]{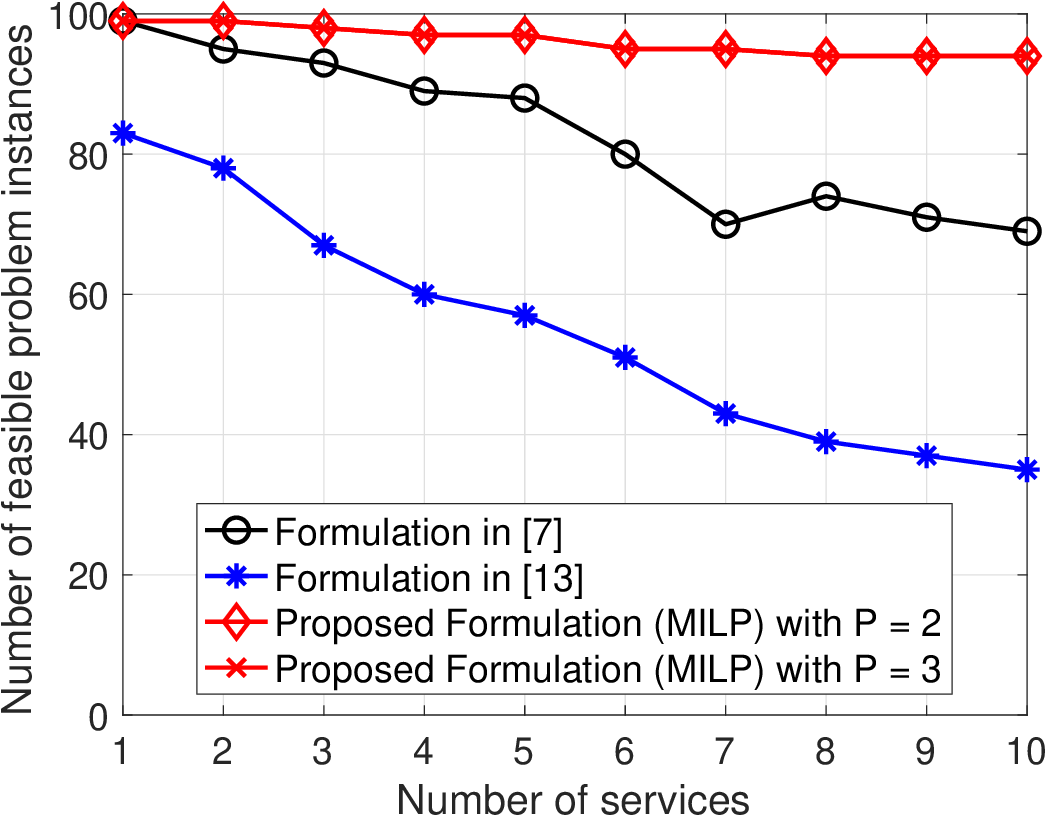}
	\caption{The number of feasible problem instances by solving \rev{formulations \eqref{milp} with $P=2, 3$} and those in  \cite{Chen2023} and \cite{Vizarreta2017}.}	
	\label{nfeasF}
\end{figure}

In this subsection, we demonstrate the effectiveness of our proposed formulation \eqref{milp} by comparing it with the two existing state-of-the-art  formulations in \cite{Chen2023} and \cite{Vizarreta2017}.
\cref{nfeasF} plots the number of feasible problem instances returned by solving \rev{the proposed formulations \eqref{milp} with $P=2$ and $P=3$ and the two formulations in \cite{Chen2023} and \cite{Vizarreta2017}.} 
For the formulation in \cite{Vizarreta2017} and the proposed formulation \eqref{milp}, since the E2E reliability constraints \eqref{zzrelcons}--\eqref{E2Ereliability2} are explicitly enforced, we solve the corresponding formulation by CPLEX and count the corresponding problem instance feasible if CPLEX can return a feasible solution; otherwise it is declared as infeasible.
As the E2E reliability constraints \eqref{zzrelcons}--\eqref{E2Ereliability2}
are not explicitly considered in the formulation in \cite{Chen2023}, the corresponding curve in Fig. \ref{nfeasF} is obtained as follows. 
We solve the formulation in \cite{Chen2023} and then substitute the obtained solution into constraints \eqref{zzrelcons}--\eqref{E2Ereliability2}: if the solution satisfies constraints \eqref{zzrelcons}--\eqref{E2Ereliability2} of all services, we count the corresponding problem instance feasible; otherwise it is infeasible. 

\rev{We first compare the performance of the proposed formulation \eqref{milp} with $P=2$ and $P=3$. 
	Intuitively, since there is more flexibility of traffic routing in  \eqref{milp}  with $P=3$, solving  \eqref{milp} with $P=3$ is likely to return a feasible solution for more problem instances than solving  \eqref{milp}  with $P=2$. However, as illustrated in \cref{nfeasF}, solving formulations  \eqref{milp}  with $P=2$ and $P=3$ provides the same number of feasible problem instances.
	These results are consistent with the previous results in \cite{Chen2021a} and highlight a useful insight that there already exists a sufficiently large flexibility of traffic routing in formulation  \eqref{milp}  with $P=2$.  
}

\rev{Next, we compare the performance of the proposed formulation \eqref{milp} with the formulations in \cite{Chen2023} and \cite{Vizarreta2017}.} 
As can be observed from \cref{nfeasF}, the proposed formulation \eqref{milp} can solve a much larger number of problem instances than that solved by the formulation in \cite{Vizarreta2017} \rev{(which is equivalent to the proposed formulation   \eqref{milp} with $P=1$)}, especially in the case where the number of services is large.
This shows the advantage of the flexibility of traffic routing in \eqref{milp}.
In addition, compared against that of solving the  formulation in \cite{Chen2023}, the number of feasible problem instances of solving the proposed formulation \eqref{milp} is also much larger, which clearly shows the advantage of our proposed formulation over that in \cite{Chen2023}, i.e., it has a guaranteed E2E reliability of the services.
In short summary, the results in Fig. \ref{nfeasF} demonstrate that, compared with those in \cite{Chen2023} and \cite{Vizarreta2017}, our proposed formulation \eqref{milp} can give a ``better'' solution.

\subsection{Efficiency of the Proposed \cCG algorithm}
\label{subsec:CGalg}

In this subsection, we compare the performance of the proposed \cCG algorithm with the exact approach using standard MILP solvers to solve formulation \eqref{milp} (called \EXACT), 
\rev{and two state-of-the-art heuristic  algorithms in the literature: the LP relaxation one-shot rounding (\LPoR) algorithm (adapted from \cite{Chowdhury2012}) and the LP relaxation dynamic rounding (\LPdR) algorithm (adapted from \cite{Chen2023}).}
To demonstrate the advantage of using the LP acceleration technique in the proposed \cCG algorithm (i.e, avoiding solving  the \MILP problem \eqref{subP} too many times by solving the \LP  problem \eqref{splp2}),
we also compare the \cCG algorithm with \vcCG where steps 1--6 and 13--14 of  \cref{alg0} are not applied. 


\begin{figure}[t]
	\centering
	\includegraphics[height=\figuresize]{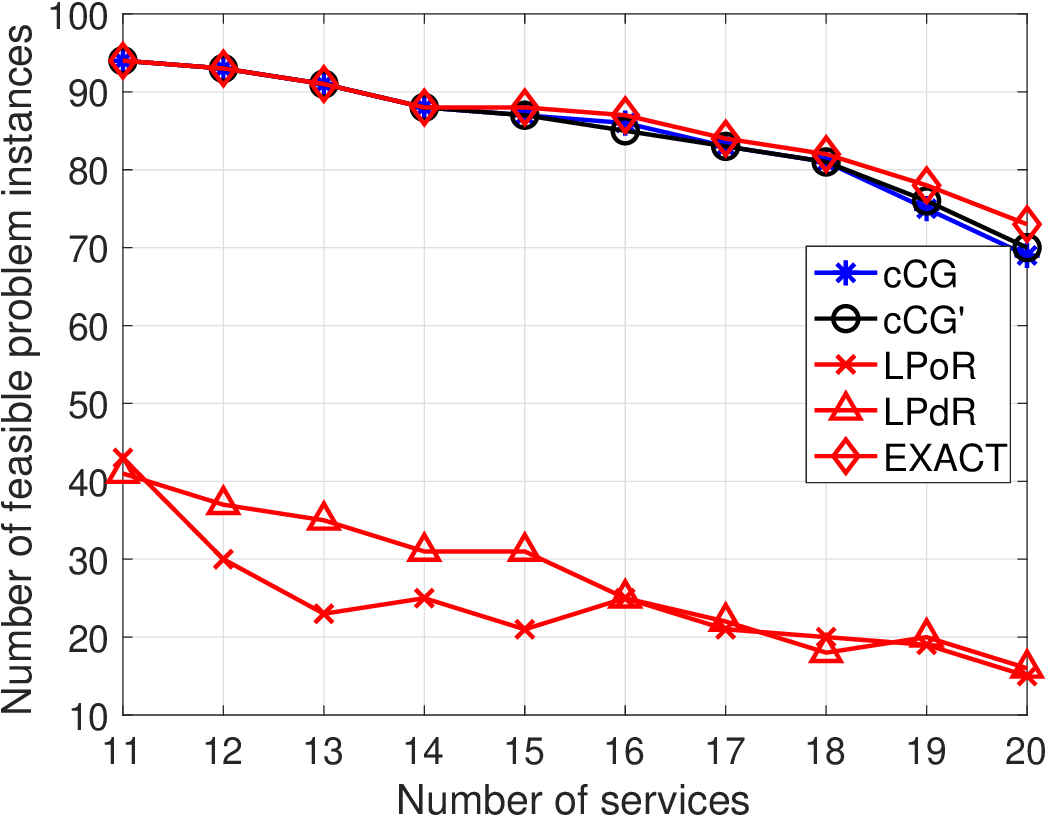}
	\caption{The number of feasible problem instances returned by \cCG, \cCG', \rev{\LPoR, \LPdR,} and \EXACT.}	
	\label{nfeasA}
\end{figure}

\begin{figure}[t]
	\centering
	\includegraphics[height=\figuresize]{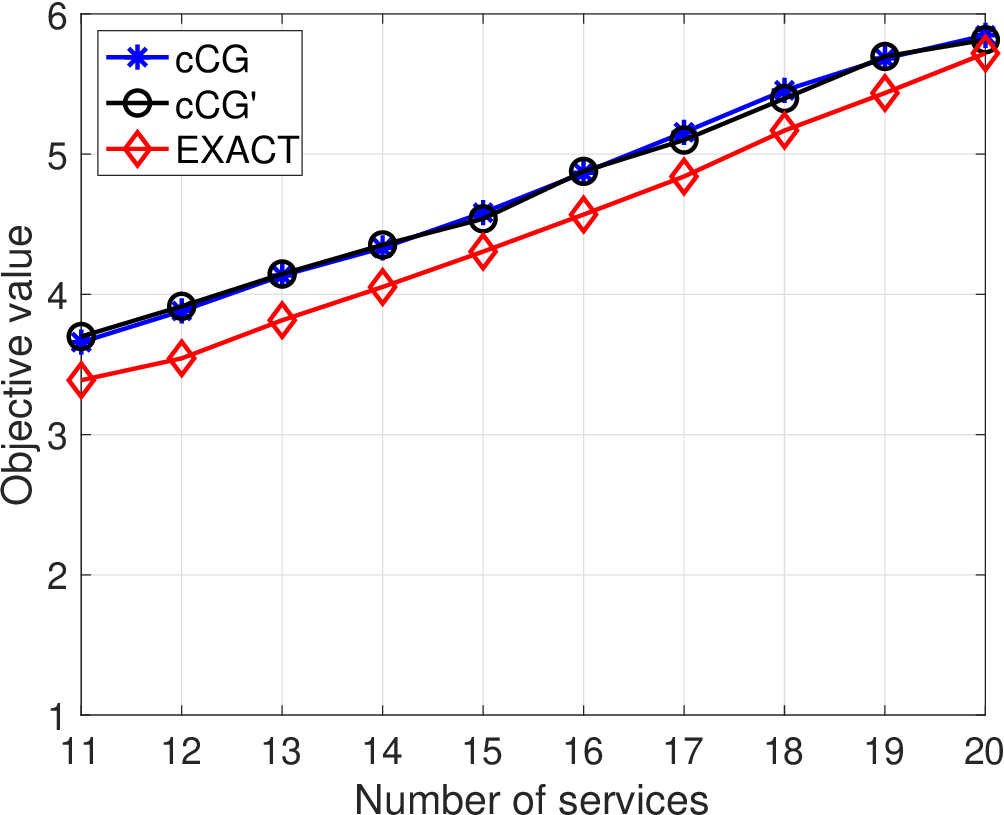}
	\caption{Average objective values returned by \cCG, \cCG', and \EXACT.}	
	\label{obj}
\end{figure}

\begin{figure}[t]
	\centering
	\includegraphics[height=\figuresize]{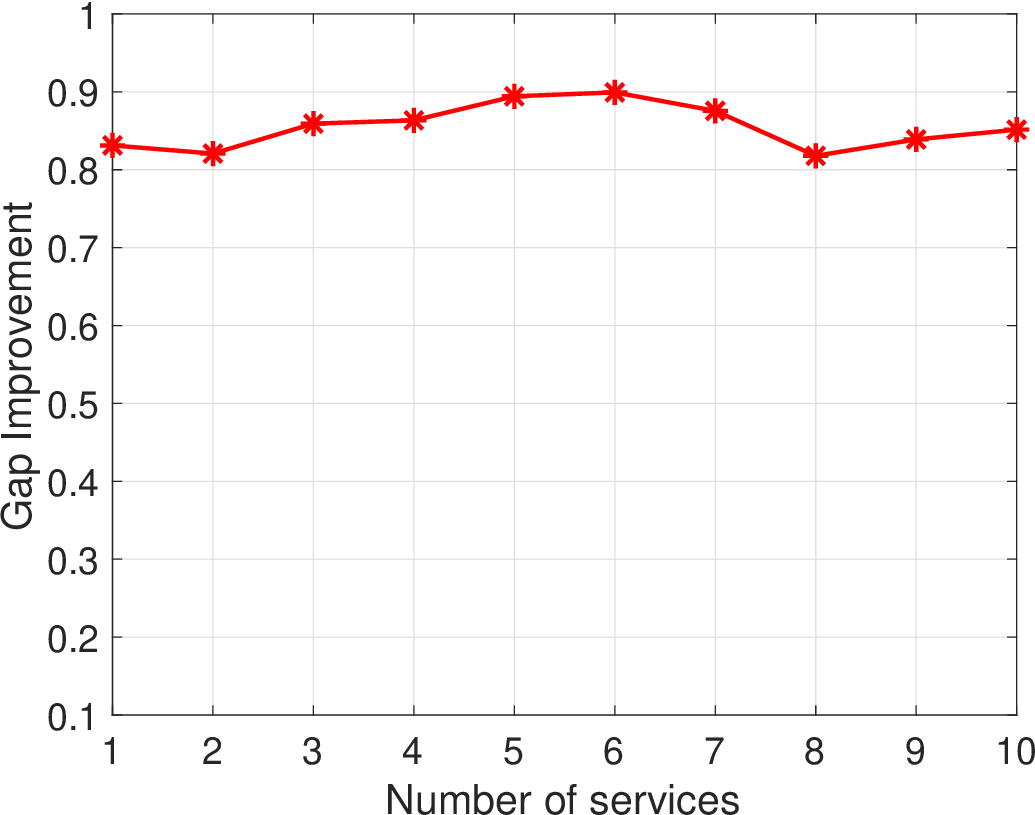}
	\caption{Average optimality gap improvement of the \LP relaxation of \eqref{flow_form} over the \LP relaxation of \eqref{milp}.}	
	\label{LPgap}
\end{figure}

\rev{\cref{nfeasA} plots the number of feasible problem instances  returned by the proposed \cCG and \vcCG, \LPoR, \LPdR, and \EXACT.}
First, from \rev{\cref{nfeasA}}, 
we observe that the blue-star curve corresponding to \cCG is almost identical to the black-circle curve corresponding to \vcCG.
This is reasonable as the \LP acceleration technique in \cCG is designed to speed up the solution procedure by avoiding solving too many \MILP problems of the form \eqref{subP} and generally does not affect the quality of the returned solution\footnote{The slight difference between the curves of \cCG and \vcCG is due to the fact that problem \eqref{subP} may have different optimal solutions and solving \eqref{splp2} may return an optimal solution of \eqref{subP}  that is different from the solution returned by directly solving \eqref{subP}; see steps 3--5 of \cref{alg0}.}.
\rev{Second,  the proposed  \cCG and \vcCG significantly outperform \LPoR and \LPdR in terms of finding feasible solutions for more problem instances.
	Indeed, when $|\K|=20$, both \cCG and \vcCG can find feasible solutions for about $70$ problem instances while \LPoR and \LPdR can only find feasible solutions for less than $20$ problem instances.
	This is due to the reason that \cCG and \vcCG explicitly take the (complex) E2E reliability and delay constraints of the services into consideration (by solving subproblem \eqref{subP}) while  \LPoR and \LPdR only consider relaxations of these constraints (obtained by relaxing binary variables into continuous variables), thereby leading to ineffective solutions.}
\rev{Finally,}
we can observe from \cref{nfeasA} that \cCG~\rev{is} capable of finding feasible solutions for almost all (truly) feasible problem instances, as \EXACT is able to {find} feasible solutions for all (truly) feasible problem instances and the {gap in} the number of feasible problem instances solved by \EXACT and \cCG is small.
\rev{The effectiveness of the proposed \cCG and \vcCG is further justified in  \cref{obj}, where we plot the average objective values returned by \cCG, \vcCG, and \EXACT.}
\rev{From \cref{obj}, we observe that \cCG and \vcCG can return  a high-quality} solution for the \NS problem, as \EXACT is able to return an optimal solution for the \NS problem and the difference of the objective values returned by \rev{\EXACT, \cCG, and \vcCG} is small.


{To gain more insight into the super performance of the proposed \cCG, we evaluate the tightness of the \LP relaxation \eqref{pLP} by comparing the relative optimality gap improvement defined by}
\begin{equation*}
	\frac{\nu(\text{RLP})-\nu(\text{LP-I})}{\nu(\text{MILP})-\nu(\text{LP-I})},
\end{equation*}
where $\nu(\text{RLP})$, $\nu(\text{LP-I})$, and $\nu(\text{MILP})$ are the optimal values of problems \eqref{pLP}, \eqref{lp}, and \eqref{milp}, respectively.
Recall that the proposed \cCG is essentially  an \LP relaxation rounding algorithm equipped the \LP relaxation \eqref{pLP} and an optimal rounding strategy; see \cref{subsec:remark}.
The results in \cref{LPgap} show that in all cases, the relative optimality gap improvement is larger than $0.8$, 
confirming that relaxation \eqref{pLP} is indeed a very tight \LP relaxation.
This is the reason why the proposed \cCG is able to find a high-quality feasible solution for the \NS problem, as can be seen in \cref{nfeasA,obj}.


\begin{figure}[t]
	\centering
	\includegraphics[height=\figuresize]{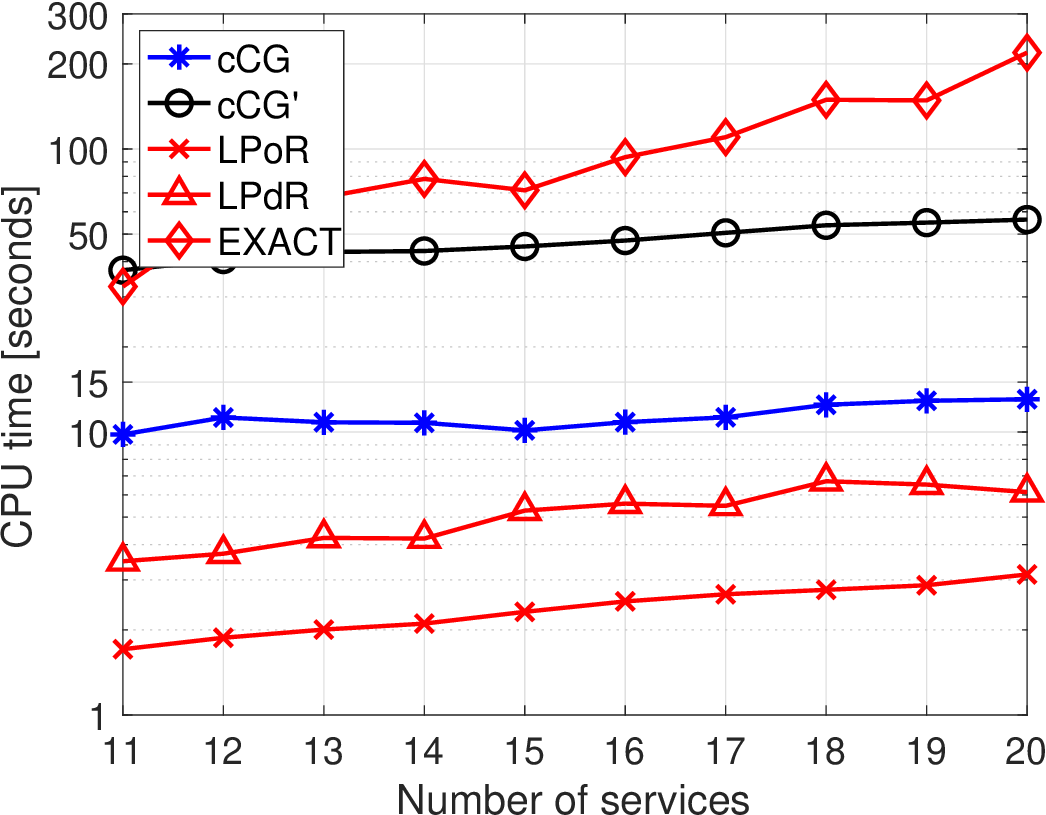}
	\caption{The average CPU time taken by \cCG, \cCG', \rev{\LPoR, \LPdR,} and \EXACT.}	
	\label{timeA}
\end{figure}

The solution efficiency comparison of the proposed \cCG and \vcCG, \rev{\LPoR, \LPdR,} and \EXACT is plotted in \cref{timeA}.
First, as expected, \cCG is much more efficient than \vcCG, which demonstrates that the \LP acceleration technique indeed can improve the performance of the proposed \cCG algorithm. 
Second, compared with \EXACT, \cCG  is significantly more efficient. 
In particular, in all cases, the average CPU times taken by \cCG are less than $15$ seconds, while the average CPU times taken by \EXACT can even be more than $200$ seconds (when $|\K|=20$).
In addition, with the increasing of the number of services, the CPU time taken by \EXACT generally grows much more rapidly than that taken by the proposed \cCG, as illustrated in \cref{timeA}.
\rev{Finally, we observe from \cref{timeA} that the difference of CPU times taken by \LPoR, \LPdR, and \vcCG is not significant; in all cases, the average CPU time is larger than $1$ second and less than $15$ seconds.
	Overall,  \LPoR and \LPdR perform better than \vcCG in terms of solution efficiency, but this comes at the cost of significantly inferior solution quality, as illustrated in \cref{nfeasA}.}

To gain more insight into the computational efficiency of the proposed algorithms \cCG and \vcCG, in \rev{\cref{nMILP,SPMILP,iter,allColumn}}, we further plot the average number of solved \MILP{s} in the form \eqref{subP}, 
\rev{the average CPU time spent in solving \eqref{subP}},
the number of iterations needed for the convergence of the first stage of \cref{alg1},
and the number of used variables $\{t^{k,c}\}$ when the first stage of \cref{alg1} is terminated, respectively.
From \cref{nMILP,SPMILP}, we can observe that compared with \vcCG, \cCG solves a much smaller number of \MILP{s}, \rev{and takes a much smaller CPU time to solve the \MILP{s}.
	This} is the main reason why \cCG is much faster than \vcCG.
From \cref{iter}, the number of iterations needed for the convergence of the first stage of \cref{alg1} is very small (e.g., smaller than $10$ in all cases).
As a result, the number of used variables $\{t^{k,c}\}$ (when the first stage of \cref{alg1} is terminated) is also very small (e.g., smaller than $90$ in all cases), as demonstrated in \cref{allColumn}.
Notice that this is also the reason why the CPU time taken by \cCG does not grow significantly with the increasing number of service $|\K|$.  

Based on the above results, we can conclude that the proposed \cCG algorithm can find a high-quality feasible solution for the \NS problem (as it is essentially an \LP~\rev{relaxation} rounding algorithm equipped with a very strong \LP relaxation and an optimal rounding strategy) while still enjoying high computational efficiency (due to the decomposition nature).


\begin{figure}[t]
	\centering
	\includegraphics[height=\figuresize]{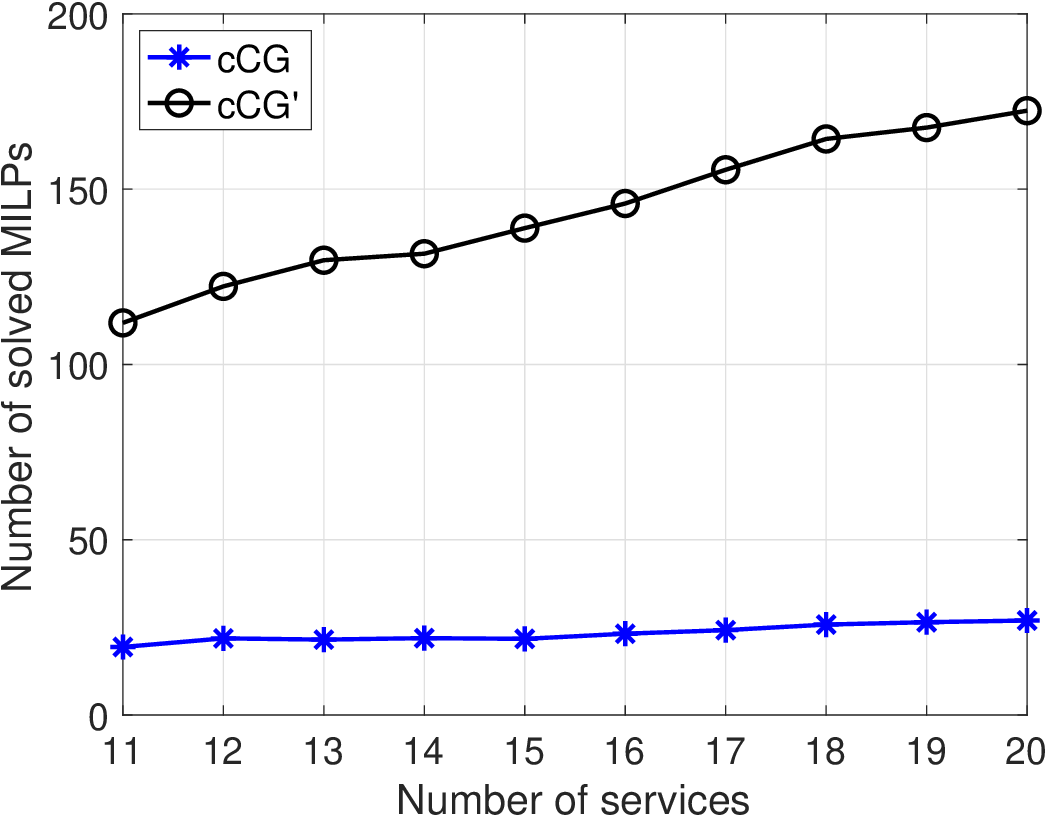}
	\caption{Average number of solved \MILP{s} by \cCG and \cCG'.}	
	\label{nMILP}
\end{figure}

\begin{figure}[t]
	\centering
	\includegraphics[height=\figuresize]{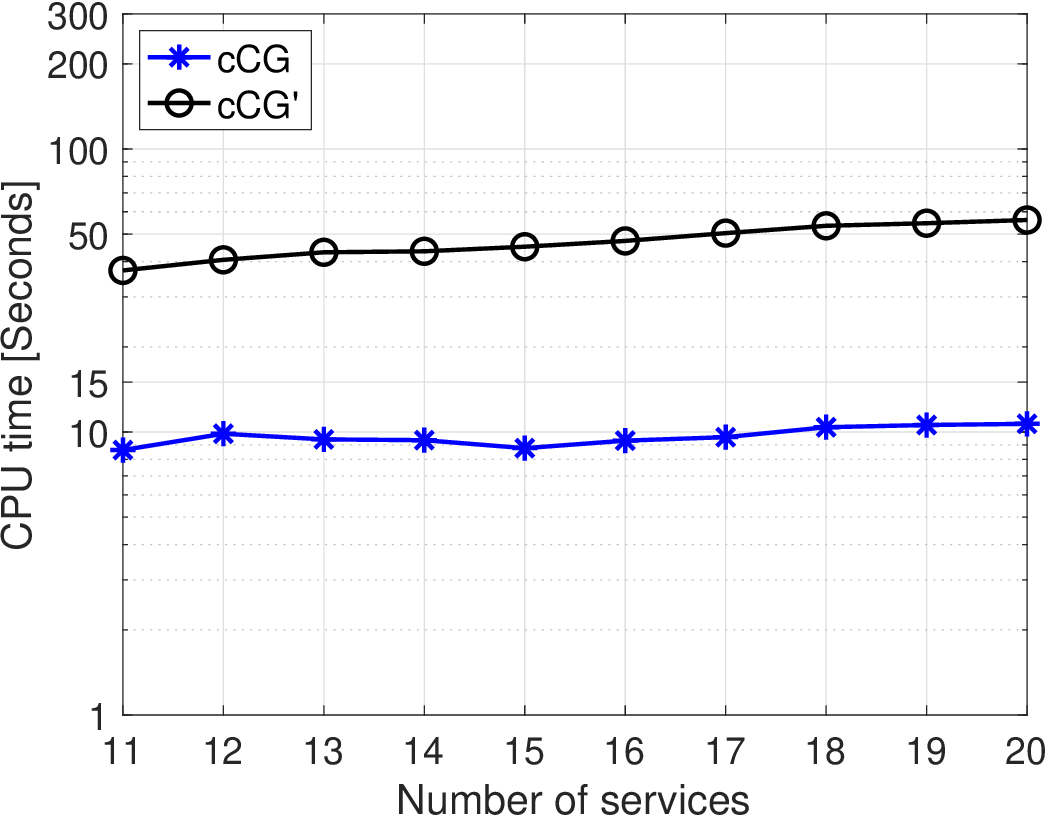}
	\caption{\rev{The average CPU time spent in solving \MILP{s} by \cCG and \vcCG.}}	
	\label{SPMILP}
\end{figure}


\begin{figure}[t]
	\centering
	\includegraphics[height=\figuresize]{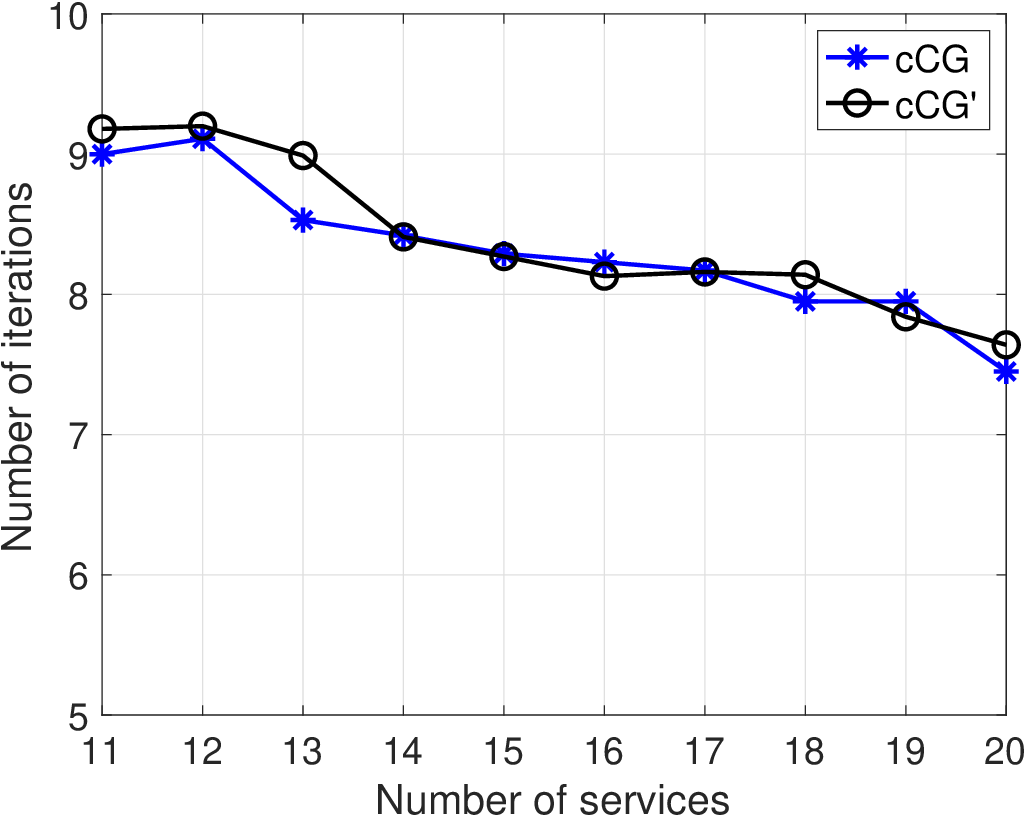}
	\caption{Average number of iterations needed by \cCG and \cCG' (for convergence).}	
	\label{iter}
\end{figure}

\begin{figure}[t]
	\centering
	\includegraphics[height=\figuresize]{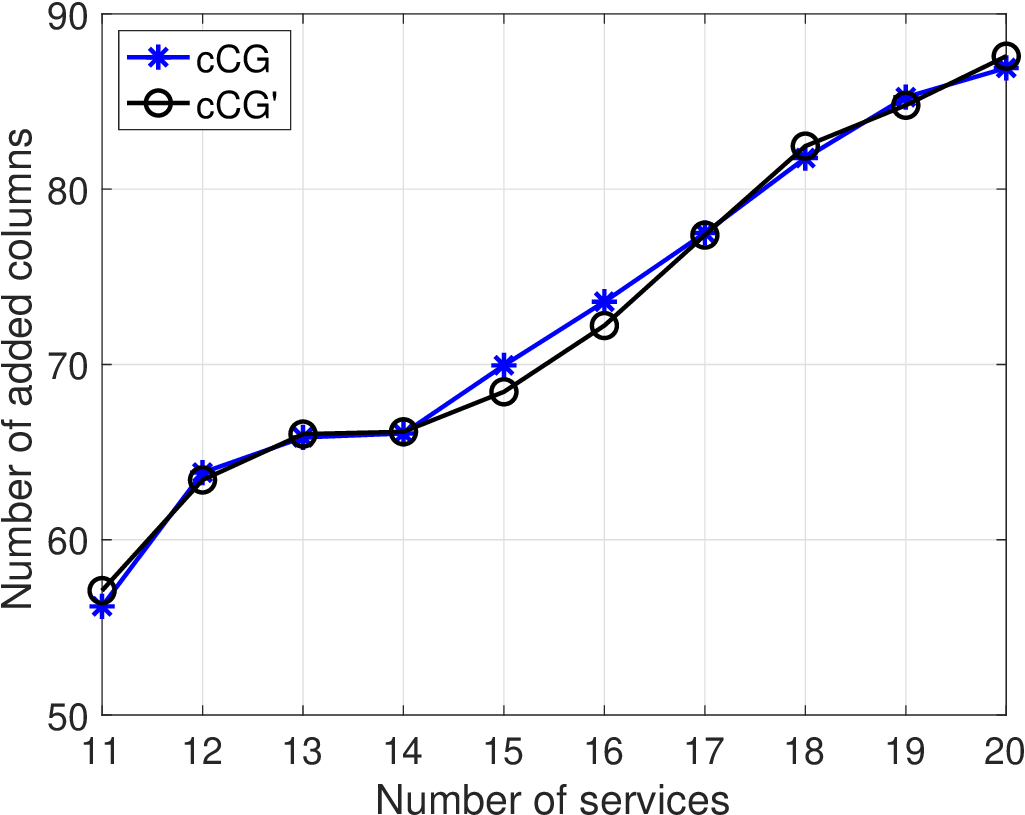}
	\caption{Average number of used columns by \cCG and \cCG'.}	
	\label{allColumn}
\end{figure}

\section{Conclusions}
\label{sect:conclusion}
In this paper, we have investigated the \NS problem that plays a crucial role in 5G and next-generation 6G communication networks.
We first developed new \MINLP and \MILP formulations for the \NS problem which minimizes a weighted sum of total power consumption and total link capacity consumption of the whole network subject to the E2E delay and reliability constraints of all services and the capacity constraints of all cloud nodes and links.
Compared with existing formulations in \cite{Chen2023,Vizarreta2017}, our proposed formulations are much more advantageous in terms of fully exploiting the flexibility of traffic routing and guaranteeing the E2E reliability of all services.
We further proposed a \cCG algorithm for solving large-scale \NS problems.
Two key features of the proposed \cCG algorithm, which makes it particularly suitable for solving \emph{large-scale} \NS problem and for finding \emph{high-quality} solutions, include: (i) it is a decomposition-based algorithm which only requires solving small-scale \LP and \MILP subproblems; 
(ii) it is essentially an \LP~\rev{relaxation} rounding algorithm equipped with a very strong \LP relaxation and an optimal rounding strategy.
Extensive computational results demonstrate the effectiveness of the proposed formulations and the efficiency of the proposed \cCG algorithm.

\rev{As observed in \cref{nMILP}, the proposed \LP acceleration technique effectively avoids solving too many \MILP subproblems  \eqref{subP}. However, as shown in \cref{timeA,SPMILP}, the main computational effort of the proposed \cCG algorithm remains concentrated on solving these \MILP subproblems.
In the future work, we shall develop efficient customized algorithm for solving these \MILP subproblems to further improve the efficiency of the proposed \cCG algorithm.}
 


\ifthenelse{\longpaper = 1}{\appendices

\section{Proof of \cref{eqformulations}}
\label{appendix1}
To prove \cref{eqformulations}, we need the follow lemma. 
\begin{lemma}\label{basiclemma}
	Suppose that \eqref{minlp} or \eqref{milp} has a feasible solution.
	Then, for \eqref{minlp} and \eqref{milp}, there  exists an optimal solution $(\bx, \by, \br, \bz, \btheta)$ such that for each $k \in \K$, $s \in \F^k \cup \{0\}$, and $p \in \CP$, 
	\begin{itemize}
		\item [(i)] if $b_{i}^{k,s}(\x) = 0 $ holds for all $i \in \I$ (i.e., the source and destination of flow $(k,s)$ are identical), then $\{(i,j) \in \CL \mid z_{ij}^{k,s,p} = 1\}$ is an empty set;
		\item [(ii)] otherwise, the links in $\{(i,j) \in \CL \mid z_{ij}^{k,s,p} = 1\}$ form an  elementary path from the source to the destination of flow $(k,s)$, namely, $i_0\rightarrow i_1 \rightarrow \cdots 
		\rightarrow i_t$ (where $t \in \mathbb{Z}_+$ with $t \geq 1$).
		Moreover, we have
		\begin{align}\label{singlepath}
			& r_{i_0i_1}^{k,s,p} = \cdots = r_{i_{t-1}i_t}^{k,s,p} =: r^{k,s,p}~\text{and} \\
			& \sum_{p \in \CP} r^{k,s,p} = 1.\label{onedatarate}
		\end{align}
	\end{itemize}
\end{lemma}
\begin{proof}
	We first prove the result for formulation \eqref{minlp}.
	Let $(\bx, \by, \br, \bz, \btheta)$ be an optimal solution of \eqref{minlp}, and let $k \in \K$, $s \in \F^k\cup \{0\}$, and $p \in \CP$.
	If $b_{i}^{k,s}(\x) = 0$ holds for all $i \in \I$,
	from \eqref{SFC1} and the classic flow decomposition result \cite[Theorem 3.5]{Ahuja1993}, the links in $\{(i,j) \in \CL \mid z_{ij}^{k,s,p} = 1\}$ form (possibly) multiple directed cycles. 
	Otherwise, from \eqref{onlyonenode}, \eqref{SFC1}, and the traditional flow decomposition result,
	the links in $\{(i,j) \in \CL \mid z_{ij}^{k,s,p} = 1\}$  form a path from the source to the destination of flow $(k,s)$ and (possibly) multiple directed cycles.
	In both cases, setting $r_{ij}^{k,s,p} = z_{ij}^{k,s,p}=0$ for the links $(i,j)$ on the circles will yield another feasible solution of \eqref{minlp} with an objective value not greater than that of  $(\bx, \by, \br, \bz, \btheta)$.
	Therefore, the newly constructed feasible solution must also be optimal.
	By recursively  applying the above argument, we get the desired result for formulation \eqref{minlp}.
	
	Next, we prove the result for formulation \eqref{milp}.
	Let $(\bx, \by, \br, \bz, \btheta)$ be an optimal solution of \eqref{milp}, and let $k \in \K$ and $s \in \F^k\cup \{0\}$.
	If $b_{i}^{k,s}(\x)=0$ holds for all $i \in \I$, from \eqref{SFC2}--\eqref{SFC5} and  the classic flow decomposition result, the flows in $\left\{ \sum_{p \in \CP} r_{ij}^{k,s,p} \right\}_{(i,j) \in \CL}$ can be decomposed into flows with positive data rates on (possibly) multiple direct cycles.
	By setting $r_{ij}^{k,s,p}= z_{ij}^{k,s,p} =0$ for  all $(i,j)\in \CL$ and $p \in \CP$, we will obtain another optimal solution of \eqref{milp} that fulfills the statement in case (i).
	Otherwise, from \eqref{SFC2}--\eqref{SFC5} and  the classic flow decomposition result,  the flows in $\left\{ \sum_{p \in \CP} r_{ij}^{k,s,p}\right\}_{(i,j) \in \CL}$ can be decomposed into  flows with positive data rates on at least one path from the source to the destination of flow $(k,s)$ and flows with positive data rates on (possibly) multiple direct cycles.
	Similarly, setting $r_{ij}^{k,s,p} =z_{ij}^{k,s,p}= 0$ for all $p \in \CP$ and all links $(i,j)$ on the circles, we will obtain another optimal solution of \eqref{milp} where $\left\{ \sum_{p \in \CP} r_{ij}^{k,s,p} \right\}_{(i,j) \in \CL}$ can be decomposed into flows with positive data rates on at least one path from the source to the destination of flow $(k,s)$.
	
	Now, for $p \in \CP$, from \eqref{SFC0} and  \eqref{SFC3}--\eqref{SFC5}, either the links in $\left\{  (i,j) \in \CL \mid r_{ij}^{k,s,p} >0 \right\}$ form a path $i_0\rightarrow i_1 \rightarrow \cdots 
	\rightarrow i_t$ from the source to the destination of flow $(k,s)$ such that \eqref{singlepath} and $r^{k,s,p} > 0$ hold, or $\left\{  (i,j) \in \CL \mid r_{ij}^{k,s,p} >0 \right\}$ is an empty set.
	For the first case, we set 
	\begin{equation*}
		\begin{aligned}
			z_{ij}^{k,s,p}=\left\{\begin{array}{ll}
				1,&\text{if~}r_{ij}^{k,s,p}>0;\\[5pt]
				0,&\text{otherwise},
			\end{array}
			\right.
		\end{aligned}
		\forall~(i,j) \in \CL, 
	\end{equation*}
	while for the second case, we set $z_{ij}^{k,s,p}= z_{ij}^{k,s,p'}$ where $p'\in \CP$ such that $\left\{  (i,j) \in \CL  \mid r_{ij}^{k,s,p'} >0 \right\}$ forms a path with positive data rates from the source to the destination of flow $(k, s)$.
	Notice that the above  $p'$ exists since  $\left\{ \sum_{p \in \CP} r_{ij}^{k,s,p}  \right\}_{(i,j) \in \CL}$ can be decomposed into flows with positive data rates on at least one path from the source to the destination of flow $(k,s)$.
	Observe that by \eqref{SFC2}, the total fraction of data rates among all paths from the source to the destination of flow $(k,s)$ should be equal to $1$, and thus \eqref{onedatarate} holds.
	Therefore, the resultant optimal solution fulfills the statement in case (ii).
\end{proof}

	\noindent {\bf Proof of \cref{eqformulations}}.
	In this proof, in order to differentiate the feasible points of the two problems, we use $ (\boldsymbol{X},\boldsymbol{Y},\boldsymbol{R},\boldsymbol{Z},\boldsymbol{\Theta}) $ and 
	$(\boldsymbol{x},\boldsymbol{y},\boldsymbol{r},\boldsymbol{z},\boldsymbol{\theta})$
	to denote the optimal solutions of problems \eqref{minlp} and \eqref{milp}, respectively.
	We prove the theorem by showing that, 
	given an optimal solution $ (\boldsymbol{X},\boldsymbol{Y},\boldsymbol{R},\boldsymbol{Z},\boldsymbol{\Theta}) $ of problem \eqref{minlp} satisfying the conditions in  \cref{basiclemma}, we can construct a feasible solution  $(\boldsymbol{x},\boldsymbol{y},\boldsymbol{r},\boldsymbol{z},\boldsymbol{\theta})$  of problem \eqref{milp} such that the two problems have the same objective value at the corresponding solutions and vice versa.
	
	Given an optimal solution $ (\boldsymbol{X},\boldsymbol{Y},\boldsymbol{R},\boldsymbol{Z},\boldsymbol{\Theta}) $ of problem \eqref{minlp} that satisfies the conditions in \cref{basiclemma},
	we construct a point  $ (\boldsymbol{x},\boldsymbol{y},\boldsymbol{r},\boldsymbol{z},\boldsymbol{\theta}) $  by setting $ \boldsymbol{x}=\boldsymbol{X}, ~\boldsymbol{y}=\boldsymbol{Y},~\boldsymbol{z}=\boldsymbol{Z},~\boldsymbol{\theta}=\boldsymbol{\Theta}$, and
	\begin{align}
		& r_{ij}^{k,s,p}= R_{ij}^{k,s,p},\nonumber                                                                                                             \\
		& \qquad\quad \forall ~(i,j) \in \mathcal{L},~k \in \mathcal{K},~s \in \mathcal{F}^k\cup \{0\}, ~p \in \mathcal{P}. \label{tmeq1}
	\end{align}
	Clearly, problems \eqref{minlp} and \eqref{milp}  have the same objective values at points $ (\boldsymbol{X},\boldsymbol{Y},\boldsymbol{R},\boldsymbol{Z},\boldsymbol{\Theta}) $ and  $ (\boldsymbol{x},\boldsymbol{y},\boldsymbol{r},\boldsymbol{z},\boldsymbol{\theta}) $, respectively, 
	and constraints \rev{\eqref{onlyonenode}--\eqref{nodecapcons}, \eqref{linkcapcons1}--\eqref{E2Ereliability2},  \eqref{delayconstraint}--\eqref{rzbounds}, and \eqref{consdelay2funs1}} hold at $ (\boldsymbol{x},\boldsymbol{y},\boldsymbol{r},\boldsymbol{z},\boldsymbol{\theta}) $.
	Since \eqref{validineq1} and \eqref{validineq2} are valid inequalities for problem \eqref{minlp}, by construction, \eqref{validineq1} and \eqref{validineq2} also hold at $ (\boldsymbol{x},\boldsymbol{y},\boldsymbol{r},\boldsymbol{z},\boldsymbol{\theta}) $.
	Therefore, we only need to show that constraints \eqref{SFC0}--\eqref{SFC5} hold at point $ (\boldsymbol{x},\boldsymbol{y},\boldsymbol{r},\boldsymbol{z},\boldsymbol{\theta}) $.
	By construction and \cref{basiclemma},  either $\{(i,j) \in \CL\mid z_{ij}^{k,s,p} = 1\}$ is an empty set, or 
	the links in $\{(i,j) \in \CL \mid z_{ij}^{k,s,p} = 1\}$ form an  elementary path from the source to the destination of flow $(k,s)$. 
	As such, \eqref{SFC0} holds at $ (\boldsymbol{x},\boldsymbol{y},\boldsymbol{r},\boldsymbol{z},\boldsymbol{\theta}) $.
	From \eqref{nonlinearcons} and $R^{k,s,p}\in [0,1]$, \eqref{linearcons} is satisfied.
	Observe that 
	\begingroup
	\allowdisplaybreaks
	\begin{align*}
		& \sum_{j: (j,i) \in \mathcal{{L}}} r_{ji}^{k, s, p} - \sum_{j: (i,j) \in \mathcal{{L}}} r_{ij}^{k, s,  p}                       \\
		&\qquad \stackrel{(a)}{=}\sum_{j: (j,i) \in \mathcal{{L}}} R^{k,s,p}Z_{ji}^{k, s,p} -  \sum_{j: (i,j) \in \mathcal{{L}}} R^{k,s,p}Z_{ij}^{k, s,p}\\
		&\qquad = R^{k,s,p} \left ( \sum_{j: (j,i) \in \mathcal{{L}}} Z_{ji}^{k, s,p} - \sum_{j: (i,j) \in \mathcal{{L}}}Z_{ij}^{k, s,p} \right ) \\
		&\qquad \stackrel{(b)}{=} R^{k,s,p} b_{i}^{k,s} (\boldsymbol{X})= R^{k,s,p} b_{i}^{k,s} (\boldsymbol{x}),
	\end{align*}
	\endgroup
	where (a) follows from \eqref{nonlinearcons} and \eqref{tmeq1} and (b) follows from \eqref{SFC1}.
	Therefore, in the last case of \eqref{SFC1} (i.e., $b_{i}^{k,s} (\boldsymbol{x})=0$), \eqref{SFC3} holds true; and in the second, third, and fourth cases of \eqref{SFC1} (i.e., $b_{i}^{k,s} (\boldsymbol{x})=x_{i}^{k,s+1},~x_{i}^{k,s+1}-x_{i}^{k,s},~-x_{i}^{k,s}$) and $0 \leq R^{k,s,p} \leq 1 $, \eqref{SFC4} and \eqref{SFC5} hold true.
	Finally, summing up the above equations over $p \in \mathcal{P}$ and using \eqref{relalambdaandx11}, we obtain
	\begingroup
	\allowdisplaybreaks
	\begin{align*}
		&   &   & \sum_{p \in \mathcal{P}}\sum_{j: (j,i) \in \mathcal{{L}}} r_{ji}^{k, s, p} - \sum_{p \in \mathcal{P}}\sum_{j: (i,j) \in \mathcal{{L}}} r_{ij}^{k, s,  p} = b_{i}^{k,s} (\boldsymbol{x}),
	\end{align*}
	\endgroup
	showing that \eqref{SFC2} also holds true at $ (\boldsymbol{x},\boldsymbol{y},\boldsymbol{r},\boldsymbol{z},\boldsymbol{\theta}) $.
	
	Next, given an optimal solution $ (\boldsymbol{x},\boldsymbol{y},\boldsymbol{r},\boldsymbol{z},\boldsymbol{\theta}) $ of  \eqref{milp} that satisfies the conditions in \cref{basiclemma},
	we construct a point $ (\boldsymbol{X},\boldsymbol{Y},\boldsymbol{R},\boldsymbol{Z},\boldsymbol{\Theta}) $ by setting $\boldsymbol{X} = \boldsymbol{x}, ~\boldsymbol{Y}=\boldsymbol{y},~\boldsymbol{Z}=\boldsymbol{z},~\boldsymbol{\Theta}=\boldsymbol{\theta}$, and
	\begin{align}
		& R_{ij}^{k,s,p}= r_{ij}^{k,s,p},\nonumber                                                                                                             \\
		& \qquad\quad \forall ~(i,j) \in \mathcal{L},~k \in \mathcal{K},~s \in \mathcal{F}^k\cup \{0\}, ~p \in \mathcal{P}. \label{tmeq2}
	\end{align}
	In addition, for $k \in \K$ and $s \in \F^k\cup \{0\}$, if case (i) of \cref{basiclemma}  happens, we set $R^{k,s,p}= 1/P$ for $p \in \CP$; otherwise, we set $R^{k,s,p} = r^{k,s,p}$ for $p \in \CP$ where $r^{k,s,p}$ is defined in  \eqref{singlepath}.
	Then it is simple to see that problems \eqref{minlp} and \eqref{milp}  have the same objective values at points $ (\boldsymbol{X},\boldsymbol{Y},\boldsymbol{R},\boldsymbol{Z},\boldsymbol{\Theta}) $ and  $ (\boldsymbol{x},\boldsymbol{y},\boldsymbol{r},\boldsymbol{z},\boldsymbol{\theta}) $, respectively, 
	and constraints \rev{\eqref{onlyonenode}--\eqref{nodecapcons}, \eqref{linkcapcons1}--\eqref{E2Ereliability2}, and \eqref{delayconstraint}--\eqref{consdelay2funs1}} hold at $ (\boldsymbol{X},\boldsymbol{Y},\boldsymbol{R},\boldsymbol{Z},\boldsymbol{\Theta}) $.
	From \cref{basiclemma}, \eqref{SFC1} also holds at  $ (\boldsymbol{X},\boldsymbol{Y},\boldsymbol{R},\boldsymbol{Z},\boldsymbol{\Theta}) $.
	We complete the proof by showing that, 
	for $k \in \K$, $s \in \F^k \cup \{0\}$, and $p \in \CP$, \eqref{relalambdaandx11} and \eqref{nonlinearcons} also hold at $ (\boldsymbol{X},\boldsymbol{Y},\boldsymbol{R},\boldsymbol{Z},\boldsymbol{\Theta}) $.
	There are two cases.
	\begin{itemize}
		\item If case (i) of \cref{basiclemma} happens, by  $R^{k,s,p}= 1/P$ for $p \in \CP$, \eqref{relalambdaandx11} holds. 
		In addition, $R_{ij}^{k,s,p}=Z_{ij}^{k,s,p}=0$ must hold for all $(i,j) \in \CL$, and thus \eqref{nonlinearcons} also holds true.
		\item Otherwise, case (ii) of \cref{basiclemma} happens. 
		Then from \eqref{singlepath} and \eqref{onedatarate}, \eqref{relalambdaandx11} and \eqref{nonlinearcons} also hold at $ (\boldsymbol{X},\boldsymbol{Y},\boldsymbol{R},\boldsymbol{Z},\boldsymbol{\Theta}) $. 
	\end{itemize}

\section{Proof of \cref{eqrelaxations}}
\label{appendix4}
The proof of Theorem 2 is divided into two main steps:
 (i) the equivalence of problems \eqref{nlp} and \eqref{lp2} and (ii) the equivalence of problems \eqref{lp2} and \eqref{lp}.
 Next, we show these two equivalences one by one.
\subsection{Equivalence of Problems \eqref{nlp} and \eqref{lp2}}
In order to prove the equivalence of problems \eqref{nlp} and \eqref{lp2}, we need the following lemma which characterizes a property of the feasible solutions of problem \eqref{nlp}.

\begin{lemma}\label{nlpproperty}
	Let $ (\bar{\bx},\bar{\by},\bar{\br},\bar{\bz},\bar{\btheta}) $ be a feasible solution of problem \eqref{nlp}.
	Then 
	\begin{equation}\label{tmeq9-1}
		\bar{r}_{ij}^{k,s} = \sum_{p \in \CP}\bar{r}_{ij}^{k,s,p}\leq 1,~  \forall ~(i,j) \in \mathcal{L},~k\in\mathcal{K},~s \in \mathcal{F}^k\cup \{0\}
	\end{equation}
	and 
	$({\bx},{\by},{\br},{\bz},{\btheta}) $ defined by  
	\begin{align}
		& {\bx} = \bar{\bx}, ~{\by}=\bar{\by},~{\btheta}=\bar{\btheta},\label{tmeq6-0} \\
		& {r}_{ij}^{k,s,p}=\left\{\begin{array}{ll}
			\bar{r}_{ij}^{k,s}, & \text{if~}p =1;              \\
			0,                                    & \text{otherwise},\label{tmeq6-1}
		\end{array}
		\right.\\
		& {r}^{k,s,p}=\left\{\begin{array}{ll}
			1, & \text{if~}p =1;              \\
			0, & \text{otherwise},\label{tmeq7-1}
		\end{array}
		\right.\\
		& {z}_{ij}^{k,s,p}=\bar{r}_{ij}^{k,s}~\text{and}~  {z}_{ij}^k=\bar{z}_{ij}^k,\nonumber\\
		&~~~~ \forall~(i,j) \in \mathcal{L}, ~k\in\mathcal{K},~s \in \mathcal{F}^k\cup \{0\}, ~p \in \mathcal{P}, \label{tmeq8-1}
	\end{align}
	is also a feasible solution of \eqref{nlp} with 
	the same objective value as that of  $ (\bar{\bx},\bar{\by},\bar{\br},\bar{\bz},\bar{\btheta}) $.
\end{lemma}
\begin{proof}
	For each $(i,j) \in \mathcal{L}$, $k\in \mathcal{K}$, and $s \in \mathcal{F}^k\cup \{0\}$, we have 
	\begin{equation*}
		\begin{aligned}
			\bar{r}_{ij}^{k,s} {=} \sum_{p \in \mathcal{P}} \bar{r}_{ij}^{k,s,p} & \stackrel{(a)}=\sum_{p \in \mathcal{P}} \bar{r}^{k,s,p} \bar{z}_{ij}^{k,s,p}  \stackrel{(b)}{\leq}\sum_{p \in \mathcal{P}} \bar{r}^{k,s,p}                \stackrel{(c)}{=} 1,
		\end{aligned}
	\end{equation*}
	where (a), (b), and (c) follow from \eqref{nonlinearcons}, $\bar{z}_{ij}^{k,s,p} \leq 1$, and \eqref{relalambdaandx11}, respectively.
	Hence, \eqref{tmeq9-1} holds.
	By the definition of $(\boldsymbol{x},\boldsymbol{y},\boldsymbol{r},\boldsymbol{z},\boldsymbol{\theta})$ in \eqref{tmeq6-0}--\eqref{tmeq8-1} and the feasibility of $(\bar{\bx},\bar{\by},\bar{\br},\bar{\bz},\bar{\btheta})$ for problem \eqref{nlp}, 
	\eqref{nlp} has the same objective values at points $ ({\bx},{\by},{\br},{\bz},{\btheta}) $ and  
	$ (\bar{\bx},\bar{\by},\bar{\br},\bar{\bz},\bar{\btheta}) $, and 
	constraints \eqref{onlyonenode}--\eqref{nodecapcons}, \eqref{relalambdaandx11}--\eqref{linkcapcons1}, \eqref{E2Ereliability2},  \eqref{delayconstraint}, \eqref{rbounds}, and \eqref{xyboundsR}--\eqref{rzboundsR} hold at $(\boldsymbol{x},\boldsymbol{y},\boldsymbol{r},\boldsymbol{z},\boldsymbol{\theta})$.
	Therefore, to prove the lemma, it suffices to show that \eqref{SFC1}, \eqref{zzrelcons}, and \eqref{consdelay2funs1} hold at point $ ({\bx},{\by},{\br},{\bz},{\btheta}) $.
	We prove them one by one below.
	
	Constraint \eqref{SFC1} holds at point  $ ({\bx},{\by},{\br},{\bz},{\btheta}) $, since
	\begingroup
	\allowdisplaybreaks
	\begin{align*}
		& \sum_{j: (j,i) \in \mathcal{{L}}} {z}_{ji}^{k, s, p} - \sum_{j: (i,j) \in \mathcal{{L}}} {z}_{ij}^{k, s,  p} \\
		& \qquad \stackrel{(a)}{=} \sum_{j: (j,i) \in \mathcal{{L}}} \sum_{p \in \mathcal{P}}\bar{r}_{ji}^{k, s,p}-\sum_{j: (i,j) \in \mathcal{{L}}} \sum_{p \in \mathcal{P}}\bar{r}_{ij}^{k, s,p} \\
		& \qquad \stackrel{(b)}{=}  \sum_{j: (j,i) \in \mathcal{{L}}} \sum_{p \in \mathcal{P}}\bar{r}^{k,s,p}\bar{z}_{ji}^{k, s,p} -                                                                     \sum_{j: (i,j) \in \mathcal{{L}}} \sum_{p \in \mathcal{P}}\bar{r}^{k,s,p}\bar{z}_{ij}^{k, s,p}  \\
		& \qquad = \sum_{p \in \mathcal{P}}\bar{r}^{k,s,p} \left ( \sum_{j: (j,i) \in \mathcal{{L}}} \bar{z}_{ji}^{k, s,p} - \sum_{j: (i,j) \in \mathcal{{L}}}\bar{z}_{ij}^{k, s,p} \right )   \\
		&\qquad  \stackrel{(c)}{=}  \sum_{p \in \mathcal{P}}\bar{r}^{k,s,p} b_{i}^{k,s} (\bar{\x})                                                              \\
		& \qquad \stackrel{(d)}{=} b_{i}^{k,s} (\bar{\x})=  b_{i}^{k,s} ({\x}),
	\end{align*}
	\endgroup
	where  (a), (b), (c), and (d) follow from \eqref{tmeq8-1}, \eqref{nonlinearcons}, \eqref{SFC1}, and \eqref{relalambdaandx11}, respectively.
	Constraint \eqref{zzrelcons} holds at point $ ({\bx},{\by},{\br},{\bz},{\btheta}) $ due to 
	\begingroup
	\allowdisplaybreaks
	\begin{align*}
		{z}_{ij}^{k,s,p}& \stackrel{(a)}{=}  \sum_{p \in \mathcal{P}}\bar{r}_{ij}^{k,s,p} \\
		& \stackrel{(b)}{=} \sum_{p \in \mathcal{P}}\bar{r}^{k,s,p}\bar{z}_{ij}^{k, s,p} \stackrel{(c)}{\leq}  \sum_{p \in \mathcal{P}}\bar{r}^{k,s,p} \bar{z}_{ij}^{k} \stackrel{(d)}{=} \bar{z}_{ij}^k=  {z}_{ij}^k,
	\end{align*}
	\endgroup
	where  (a), (b), (c), and (d) follow from \eqref{tmeq8-1}, \eqref{nonlinearcons}, \eqref{zzrelcons}, and \eqref{relalambdaandx11}, respectively.
	Finally, \eqref{consdelay2funs1} also holds, since 
	\begin{equation*}
		\begin{aligned}
			 \sum_{(i,j) \in \mathcal{{L}}}  d_{ij}  {z}_{ij}^{k, s,p}
			& \stackrel{(a)}{=}     \sum_{(i,j) \in \mathcal{{L}}}  d_{ij} \sum_{p \in \mathcal{P}}  \bar{r}_{ij}^{k, s,p}                 \\
			& \stackrel{(b)}{=}    \sum_{(i,j) \in \mathcal{{L}}}  d_{ij} \sum_{p \in \mathcal{P}}  \bar{r}^{k,s,p} \bar{z}_{ij}^{k, s,p} \\
			& =     \sum_{p \in \mathcal{P}}  \bar{r}^{k,s,p}  \sum_{(i,j) \in \mathcal{{L}}}  d_{ij} \bar{z}_{ij}^{k, s,p}                                      \\
			& \stackrel{(c)}{\leq}  \sum_{p \in \mathcal{P}}  \bar{r}^{k,s,p} \bar{\theta}^{k,s}                                      \stackrel{(d)}{=}  \bar{\theta}^{k,s}                                                                 =    {\theta}^{k,s},
		\end{aligned}
	\end{equation*}
	where (a), (b), (c), and (d) follow from \eqref{tmeq8-1}, \eqref{nonlinearcons}, \eqref{consdelay2funs1}, and \eqref{relalambdaandx11}, respectively.
	The proof is complete.
\end{proof}

\cref{nlpproperty} immediately implies that, if problem \eqref{nlp} is feasible,    
then it must have an optimal solution $(\boldsymbol{x},\boldsymbol{y},\boldsymbol{r},\boldsymbol{z},\boldsymbol{\theta})$ such that
\begin{align}
	& r^{k,s,1} = 1, ~\forall ~k\in\mathcal{K},~s \in \mathcal{F}^k\cup \{0\},\label{tmpeq1}\\
	& r^{k,s,p}=0,~ \forall~k\in\mathcal{K},~s \in \mathcal{F}^k\cup \{0\}, ~p \in \mathcal{P}\backslash\{1\},\\
	& r_{ij}^{k,s,p}=0,\nonumber \\
	& \quad \forall~(i,j) \in \mathcal{L},~k\in\mathcal{K},~s \in \mathcal{F}^k\cup \{0\}, ~p \in \mathcal{P}\backslash\{1\},\label{tmpeq2}\\
	& z_{ij}^{k,s,p}=r_{ij}^{k,s,1} \leq 1,\nonumber\\
	& {\qquad\quad \forall ~(i,j) \in \mathcal{L},~k\in\mathcal{K},~s \in \mathcal{F}^k\cup \{0\}, ~p \in \mathcal{P}.}\label{tmpeq3}
\end{align}
Therefore, we can introduce new variables $r_{ij}^{k,s}:=r_{ij}^{k,s,1}$, and  remove all the other variables $\{r^{k,s,p}\}$, $\{ z_{ij}^{k,s,p} \}$, and  $\{r_{ij}^{k,s,p}\}$ (using Eqs. \eqref{tmpeq1}--\eqref{tmpeq3}) from problem \eqref{nlp}.
In addition, constraints \eqref{relalambdaandx11} and \eqref{nonlinearcons} can be removed from the problem and  constraints \eqref{SFC1}, \eqref{linkcapcons1}, \eqref{zzrelcons},   \eqref{consdelay2funs1}, and \eqref{rzboundsR} can be replaced by \eqref{mediacons2-2}, \eqref{linkcapcons2}, \eqref{zzrelcons2}, \eqref{consdelay2funs2}, and \eqref{NLPrzbounds}, respectively. 
Therefore, problem \eqref{nlp} reduces to problem \eqref{lp2}.\\[5pt]
%
\subsection{Equivalence of Problems \eqref{lp} and \eqref{lp2}}
In this part, we prove that problems \eqref{lp} and \eqref{lp2} are equivalent.
To do this, we first show that constraint \eqref{SFC0} can  be equivalently removed from problem \eqref{lp}.
Let (LP-I') be the problem obtained by removing constraint \eqref{SFC0} from problem \eqref{lp}. 
Observe that, for a feasible solution $(\x, \y, \r, \z, \btheta)$ of problem \eqref{lp} or (LP-I'), if $z_{ij}^{k,s,p} > r_{ij}^{k,s,p}$ holds for some $(i,j) \in \CL$, $k \in \K$, $s \in \F^k\cup \{0\}$, and $p \in \CP$, then setting $z_{ij}^{k,s,p} := r_{ij}^{k,s,p}$ yields another feasible solution with the same objective value. 
This result is formally stated below.
\begin{lemma}\label{obser}
	Setting
	\begin{equation}\label{optcond}
		z_{ij}^{k,s,p} =r_{ij}^{k,s,p}, ~\forall~(i,j)\in \mathcal{L},~k\in \mathcal{K},~s\in \mathcal{F}^k\cup \{0\},~p \in \mathcal{P}
	\end{equation}
	in problem \eqref{lp} or {\rm (LP-I')} does not change its optimal value.
\end{lemma}
\begin{lemma}
	\label{lppropery1}
	Problem \eqref{lp} and {\rm (LP-I')} are equivalent, that is,  there exists an optimal solution $(\boldsymbol{x},\boldsymbol{y},\boldsymbol{r},\boldsymbol{z},\boldsymbol{\theta})$ of problem (LP-I') such that 
	\begin{align}
		& \sum_{j: (i,j) \in \mathcal{{L}}} z_{ij}^{k, s,  p} \left(= \sum_{j: (i,j) \in \mathcal{{L}}} r_{ij}^{k, s,  p}\right) \leq 1, \nonumber\\
		& \qquad\qquad \forall~i \in \mathcal{I},~k \in \mathcal{K},~s \in \mathcal{F}^k\cup \{0\},~ p \in \mathcal{P}.\label{optcond2}
	\end{align}
\end{lemma}
\begin{proof}
	Let $(\bx, \by, \br, \bz, \btheta)$ be an optimal solution of problem (LP-I') that satisfies \eqref{optcond}. 
	From \eqref{SFC2}--\eqref{SFC5} and the traditional flow decomposition result, 
	the traffic flow in $\{ r_{ij}^{k,s,p}\}_{(i,j) \in \CL}$ can be decomposed into flows with positive data rates on  (possibly) multiple paths from the source to the destination of flow $(k,s)$, and  flows with positive data rates on (possibly) multiple directed cycles.
	Setting $r_{ij}^{k,s,p} =z_{ij}^{k,s,p}= 0$ for the links $(i,j)$ on all circles yields another optimal solution that satisfies \eqref{optcond2}.
	The proof is complete.
\end{proof}

Similar to \cref{nlpproperty}, we provide the following lemma to  characterize a property of the feasible solutions of problem (LP-I').
\begin{lemma}
	\label{lpproperty2}
	Let $ (\bar{\bx},\bar{\by},\bar{\br},\bar{\bz},\bar{\btheta}) $ be a feasible solution of problem {\rm (LP-I')} satisfying \eqref{optcond}--\eqref{optcond2}.
	Then \eqref{tmeq9-1} holds and the point $({\bx},{\by},{\br},{\bz},{\btheta}) $ defined by  \eqref{tmeq6-0}, \eqref{tmeq6-1}, and \eqref{tmeq8-1}
	is also a feasible solution of (LP-I') with 
	the same objective value as that of  $ (\bar{\bx},\bar{\by},\bar{\br},\bar{\bz},\bar{\btheta}) $.
\end{lemma}
\begin{proof}
	Eq. \eqref{tmeq9-1} directly follows from \eqref{optcond2}.
	From the feasibility of $(\bar{\bx},\bar{\by},\bar{\br},\bar{\bz},\bar{\btheta})$ for problem (LP-I') and the definition of $(\boldsymbol{x},\boldsymbol{y},\boldsymbol{r},\boldsymbol{z},\boldsymbol{\theta})$ in \eqref{tmeq6-0}, \eqref{tmeq6-1}, and \eqref{tmeq8-1},
	(LP-I') has the same objective values at points $ ({\bx},{\by},{\br},{\bz},{\btheta}) $ and  
	$ (\bar{\bx},\bar{\by},\bar{\br},\bar{\bz},\bar{\btheta}) $, and 
	constraints \eqref{onlyonenode}--\eqref{nodecapcons}, \eqref{linkcapcons1}, \eqref{E2Ereliability2}, \eqref{delayconstraint}, and  \eqref{linearcons}--\eqref{rzboundsR}  hold at $(\boldsymbol{x},\boldsymbol{y},\boldsymbol{r},\boldsymbol{z},\boldsymbol{\theta})$.
	From \eqref{validineq1}, \eqref{tmeq9-1}, and \eqref{tmeq8-1}, constraint \eqref{zzrelcons} holds at $ ({\bx},{\by},{\br},{\bz},{\btheta}) $; and from \eqref{validineq2}, \eqref{tmeq9-1}, and \eqref{tmeq8-1}, constraint \eqref{consdelay2funs1} also holds true.
	This proves that $(\boldsymbol{x},\boldsymbol{y},\boldsymbol{r},\boldsymbol{z},\boldsymbol{\theta})$ is  also a feasible solution of (LP-I') with the same objective value as that of $(\bar{\bx},\bar{\by},\bar{\br},\bar{\bz},\bar{\btheta})$.
\end{proof}
\cref{lpproperty2} implies that, if relaxation (LP-I') is feasible,    
then it must have an optimal solution $(\boldsymbol{x},\boldsymbol{y},\boldsymbol{r},\boldsymbol{z},\boldsymbol{\theta})$ such that \eqref{tmeq9-1}, \eqref{tmpeq2}, \eqref{tmpeq3}, and  \eqref{optcond} hold. 
Therefore, we can introduce new variables $r_{ij}^{k,s}:=r_{ij}^{k,s,1}$, and remove variables $\{r^{k,s,p}\}$, $\{ z_{ij}^{k,s,p} \}$, and  $\{r_{ij}^{k,s,p}\}$ (using Eqs. \eqref{tmpeq2}, \eqref{tmpeq3}, and  \eqref{optcond}) from problem (LP-I').
In addition, constraints \eqref{SFC2}--\eqref{SFC5}, \eqref{linkcapcons1}, \eqref{zzrelcons}, \eqref{consdelay2funs1}, and \eqref{rzboundsR} can be replaced by \eqref{mediacons2-2}, \eqref{linkcapcons2}, \eqref{zzrelcons2},  \eqref{consdelay2funs2}, and \eqref{NLPrzbounds}, respectively, and
constraints \eqref{linearcons} and \eqref{validineq1}--\eqref{validineq2} can be removed from the problem.  
Then, problem (LP-I') reduces to problem \eqref{lp2}. 
This, together with \cref{lppropery1}, implies that problem \eqref{lp} is equivalent to problem \eqref{lp2}.

}{}

\def\longreference{0}
\ifthenelse{\longreference = 1}{
	\newcommand{\JCST}{{IEEE communications surveys \& tutorials}\xspace}
	\newcommand{\JTNSM}{{IEEE Transactions on Network and Service Management}\xspace}
	\newcommand{\JTN}{{IEEE/ACM Transactions on Networking}\xspace}
	\newcommand{\JSAC}{IEEE Journal on Selected Areas in Communications\xspace}
	\newcommand{\JTSP}{IEEE Transactions on Signal Processing\xspace}
	\newcommand{\JCCR}{Computer Communication Review\xspace}
	\newcommand{\JCE}{Chinese Journal of Electronics\xspace}
	\newcommand{\JTWC}{IEEE Transactions on Wireless Communications\xspace}
	\newcommand{\JTMC}{IEEE Transactions on Mobile Computing\xspace}
	\newcommand{\JIJCS}{International Journal of Communication Systems\xspace}
	\newcommand{\JOR}{Operations Research\xspace}
	\newcommand{\JMP}{Mathematical Programming\xspace}
	\newcommand{\JICM}{IEEE Communications Magazine\xspace}
	\newcommand{\JICST}{IEEE Communications Surveys \& Tutorials}
	\newcommand{\JMPC}{Mathematical Programming Computation}
	\newcommand{\JTII}{IEEE Transactions on Industrial Informatics}
	\newcommand{\JACS}{ACM Computing Surveys}

	\newcommand{\PROC}{Proceedings of\xspace}
	
	\newcommand{\JAN}{{January}\xspace}
	\newcommand{\FEB}{Feburary\xspace}
	\newcommand{\MAR}{March\xspace}
	\newcommand{\APR}{April\xspace}
	\newcommand{\MAY}{May\xspace}
	\newcommand{\JUN}{June\xspace}
	\newcommand{\JUL}{July\xspace}
	\newcommand{\AUG}{August\xspace}
	\newcommand{\SEP}{September\xspace}
	\newcommand{\OCT}{October\xspace}
	\newcommand{\NOV}{November\xspace}
	\newcommand{\DEC}{December\xspace}
}
{
	\newcommand{\JCST}{{IEEE Commun. Surv. Tut.}\xspace}
	\newcommand{\JTNSM}{{IEEE Trans. Netw. Service Manag.}\xspace}
	\newcommand{\JTN}{{IEEE/ACM Trans. Netw.\xspace}}
	\newcommand{\JSAC}{IEEE J. Sel. Areas Commun.\xspace}
	\newcommand{\JTSP}{IEEE Trans. Signal Process.\xspace}
	\newcommand{\JCCR}{Comput. Commun. Rev.\xspace}
	\newcommand{\JCE}{Chin. J. Electron.\xspace}
	\newcommand{\JTWC}{IEEE Trans. Wirel. Commun.\xspace}
	\newcommand{\JTMC}{IEEE Trans. Mob. Comput.\xspace}
	\newcommand{\JIJCS}{Int. J. Commun. Syst.\xspace}
	\newcommand{\JOR}{Oper. Res.\xspace}
	\newcommand{\JMP}{Math. Program.\xspace}
	\newcommand{\JICM}{IEEE Commun. Mag.\xspace}
	\newcommand{\JICST}{IEEE Commun. Surv. Tutor.}
	\newcommand{\JMPC}{Math. Program. Comput.}
	\newcommand{\JTII}{IEEE Trans. Ind. Inform.}
	\newcommand{\JACS}{ACM Comput. Surv.}

	\newcommand{\PROC}{Proc.\xspace}
	
	\newcommand{\JAN}{Jan.\xspace}
	\newcommand{\FEB}{Feb.\xspace}
	\newcommand{\MAR}{Mar.\xspace}
	\newcommand{\APR}{Apr.\xspace}
	\newcommand{\MAY}{May\xspace}
	\newcommand{\JUN}{Jun.\xspace}
	\newcommand{\JUL}{Jul.\xspace}
	\newcommand{\AUG}{Aug.\xspace}
	\newcommand{\SEP}{Sep.\xspace}
	\newcommand{\OCT}{Oct.\xspace}
	\newcommand{\NOV}{Nov.\xspace}
	\newcommand{\DEC}{Dec.\xspace}
	
}

\bibliographystyle{IEEEtran}

\end{document}